\documentclass[11pt]{article}
\usepackage[margin=1in]{geometry}
\usepackage[colorlinks=true,
            linkcolor=blue,
            citecolor=black,
            urlcolor=blue]{hyperref}

\usepackage{setspace}

\usepackage{algorithm}
\usepackage{algorithmicx}
\usepackage{algpseudocode}

\usepackage{fixmath}
\usepackage{bm}
\usepackage{amsbsy}
\usepackage{color}
\usepackage{verbatim}
\usepackage{multirow}
\usepackage{amssymb}
\usepackage{array}
\usepackage{mathtools}
\usepackage{amsmath}
\usepackage{booktabs}
\usepackage{graphicx}
\usepackage{tikz}
\usetikzlibrary{positioning}
\usepackage{thm-restate}
\usepackage{subcaption}
\usepackage{listings}
\usepackage{url}

\usepackage{acro}

\newcommand{\size}[1]{\ensuremath{|#1|}}

\newcommand{\Size}[1]{\ensuremath{\left|#1\right|}}
\newcommand{\Ceil}[1]{\ensuremath{\left\lceil#1\right\rceil}}
\newcommand{\Floor}[1]{\ensuremath{\left\lfloor#1\right\rfloor}}

\newcommand{\lra}[1]{\ensuremath{(#1)}}

\newcommand{\lrA}[1]{\ensuremath{\left(#1\right)}}
\newcommand{\lrB}[1]{\ensuremath{\left[#1\right]}}
\newcommand{\lrC}[1]{\ensuremath{\left\{#1\right\}}}

\newcommand{\EEE}[1]{\ensuremath{\mathbb{E}\lrB{#1}}}

\newcommand{\qand}{\quad\text{ and }\quad}

\def\B{\mathcal{B}}

\def\F{\mathcal{F}}

\def\O{\mathcal{O}}
\def\P{\mathcal{P}}
\def\R{\mathcal{R}}
\def\T{\mathcal{T}}

\def\W{\mathcal{W}}

\def\OPT{\mbox{OPT}}

\def\SOL{\mbox{SOL}}

\def\vertexnum{n}
\def\rnum{m}
\def\vehiclenum{h}

\allowdisplaybreaks

\DeclareAcronym{mDaRP}{
short = mDaRP,
long = multi-vehicle dial-a-ride problem
}
\DeclareAcronym{mTSP}{
short = mTSP,
alt = multi-depot TSP,
long = multi-depot traveling salesman problem
}
\DeclareAcronym{CVRP}{
short = CVRP,
long = capacitated vehicle routing problem
}
\DeclareAcronym{ITP}{
short = ITP,
long = iterated tour partition
}
\DeclareAcronym{DaRP}{
short = DaRP,
long = dial-a-ride problem
}
\DeclareAcronym{SCP}{
short = SCP,
long = stacker-crane problem
}
\DeclareAcronym{TSP}{
short = TSP,
long = traveling salesman problem
}
\usepackage{natbib}
 \bibpunct[, ]{(}{)}{,}{a}{}{,}

\usepackage{amsthm}
\newtheorem{theorem}{Theorem}
\newtheorem{lemma}{Lemma}
\newtheorem{corollary}{Corollary}

\usepackage{authblk}

\title{Improved Approximations for Dial-a-Ride Problems}

\author[1,2]{Jingyang Zhao}
\author[1]{Mingyu Xiao}

\affil[1]{University of Electronic Science and Technology of China, Chengdu, China}
\affil[2]{Kyung Hee University, Yongin-si, South Korea}

\date{}

\begin{document}

\maketitle

\begin{abstract}
The multi-vehicle dial-a-ride problem (mDaRP) is a fundamental vehicle routing problem with pickups and deliveries, widely applicable in ride-sharing, economics, and transportation.
Given a set of $n$ locations, $h$ vehicles of identical capacity $\lambda$ located at various depots, and $m$ ride requests each defined by a source and a destination, the goal is to plan \emph{non-preemptive} routes that serve all requests while minimizing the total travel distance, ensuring that no vehicle carries more than $\lambda$ passengers at any time.
The best-known approximation ratio for the mDaRP remains $\O(\sqrt{\lambda}\log m)$.

We propose two simple algorithms: the first achieves the same approximation ratio of $\O(\sqrt{\lambda}\log m)$ with improved running time, and the second attains an approximation ratio of $\O(\sqrt{\frac{m}{\lambda}})$.
A combination of them yields an approximation ratio of $\O(\sqrt[4]{n}\log^{\frac{1}{2}}n)$ under $m=\Theta(n)$. Moreover, for the case $m\gg n$, by extending our algorithms, we derive an $\O(\sqrt{n\log n})$-approximation algorithm, which also improves the current best-known approximation ratio of $\O(\sqrt{n}\log^2n)$ for the classic (single-vehicle) DaRP, obtained by Gupta~\emph{et al.} (ACM Trans. Algorithms, 2010).
\end{abstract}

\maketitle

\section{Introduction}
In the \ac{mDaRP}~\citep{DBLP:journals/pvldb/LuoFDG23}, we are given a complete graph $G=(V,E,w)$ with $\size{V}=n$ locations, a weight function $w:E\to\mathbb{R}^+$, a set of $h$ identical vehicles $K\subseteq V$, and a (multi-)set of $m$ requests $\R=\{(s_i,t_i)\}_{i=1}^{m}\subseteq V\times V$.
All vehicles have the same capacity of $\lambda\in\mathbb{N}^+$.
The goal is to plan routes for vehicles to satisfy all requests, while minimizing the total travel distance. Each vehicle is not required to return to its starting point after completing its assigned requests. 
Moreover, for each vehicle, the number of requests served on its route must never exceed its capacity at any point in time.
We assume that the edge-weight function $w$ is a \emph{metric}, satisfying symmetry and the triangle inequality; that is, for any $a,b,c\in V$, $w(a, a)=0$ and $w(a,c)=w(c,a)\leq w(a,b)+w(b,c)$.
When $h=1$, the \ac{mDaRP} is known as the classic \ac{DaRP}~\citep{stein1978scheduling,psaraftis1980dynamic,DBLP:journals/ior/Cordeau06,ho2018survey}.

In the \ac{mDaRP}, each request $(s_i,t_i)$ corresponds to an item that needs to be transported from the source $s_i$ to the designated destination $t_i$. 
Based on the properties of requests, two versions are commonly considered in the literature: the \emph{preemptive} version and the \emph{non-preemptive} version.
In the preemptive version, after picking up an item at its source, the vehicle is allowed to temporarily drop off the item at some intermediate vertices before delivering it to its destination.  
In the non-preemptive version, once the vehicle picks up an item at its source, it must carry the item until it reaches its destination without any intermediate drop-offs.

As in the previous studies~\citep{DBLP:journals/talg/GuptaHNR10,DBLP:journals/pvldb/LuoFDG23}, we study the non-preemptive \ac{mDaRP}.

\subsection{Related Work}
We briefly review approximation algorithms for the \ac{DaRP} and the \ac{mDaRP}, respectively.

\textbf{DaPR.}
When $\lambda=1$, the \ac{DaRP} is also known as the \ac{SCP}~\citep{DBLP:conf/focs/FredericksonHK76,DBLP:journals/siamcomp/FredericksonHK78}.  
Furthermore, if the source and destination of each request are the same, the \ac{SCP} further reduces to the \ac{TSP}.
For the \ac{TSP}, \citet{christofides1976worst} proposed a $\frac{3}{2}$-approximation algorithm.
Recently, \citet{DBLP:conf/stoc/KarlinKG21,DBLP:conf/ipco/KarlinKG23} improved the ratio to $\frac{3}{2}-10^{-36}$.  
For the \ac{SCP}, \citet{DBLP:conf/focs/FredericksonHK76,DBLP:journals/siamcomp/FredericksonHK78} proposed a well-known $\frac{9}{5}$-approximation algorithm, which remains the best-known to date.

When all requested objects are identical, the \ac{DaRP} becomes the $\lambda$-delivery \ac{TSP}~\citep{anily1999approximation,DBLP:journals/siamcomp/ChalasaniM99,DBLP:journals/siamcomp/CharikarKR01}.
Furthermore, when all request sources coincide with the vehicle depot, it becomes the well-known \ac{CVRP}~\citep{dantzig1959truck}.
Approximation algorithms for the \ac{CVRP} can be found in \citep{HaimovichK85,altinkemer1987heuristics,blauth2022improving,uncvrp}.
Although the \ac{DaRP} is a fundamental vehicle routing problem in Operations Research~\citep{ho2018survey}, research on approximation algorithms is relatively limited.

When $\lambda \ge 1$, the $\frac{9}{5}$-approximation algorithm for the case $\lambda = 1$~\citep{DBLP:journals/siamcomp/FredericksonHK78} directly implies an $\O(\lambda)$-approximation algorithm.
\citet{DBLP:conf/focs/CharikarR98} were the first to propose novel approximation algorithms. 
For the non-preemptive case, they obtain an approximation ratio of $\O(\sqrt{\lambda}\log\vertexnum\log\log\vertexnum)$.  
For the preemptive case, they obtain an approximation ratio of $\O(\log\vertexnum\log\log\vertexnum)$.  
Their approaches are based on the \emph{tree-embedding} technique in~\citep{DBLP:conf/stoc/Bartal98}, which randomly embeds general metric spaces into tree metrics with an expected distortion of $\O(\log\vertexnum\log\log\vertexnum)$.  
This distortion bound was later improved to $\O(\log\vertexnum)$~\citep{DBLP:journals/jcss/FakcharoenpholRT04}, which is asymptotically tight.
Hence, the corresponding approximation ratios improve to $\O(\sqrt{\lambda}\log\vertexnum)$ for the non-preemptive case, and $\O(\log\vertexnum)$ for the preemptive case.

Later, \citet{DBLP:conf/esa/GuptaHNR07,DBLP:journals/talg/GuptaHNR10} introduced the \emph{$\lambda$-Steiner forest} technique~\citep{DBLP:journals/algorithmica/SegevS10}, which greedily covers all requests using \emph{minimum-ratio} $\lambda$-Steiner forests, as the greedy algorithm for the $\lambda$-set cover problem~\citep{williamson2011design}.  
They mainly focused on the non-preemptive case and achieved an approximation ratio of $\O(\min\{\sqrt{\lambda},\sqrt{ n}\}\cdot\log^2 n)$. 
They also showed that their techniques extend to more general \ac{DaRP} settings, including cases where some requests require more than one unit of vehicle capacity.

{
Approximation algorithms for the non-preemptive \ac{DaRP} on paths have also been studied.
\citet{krumke2000approximation} proposed a 3-approximation algorithm, which was later improved to $2.5$ by \citet{hu2009approximation}.
}

In the literature, there exist few results on the inapproximability for the \ac{DaRP}.
For the \ac{TSP}, \citet{karpinski2015new} proved that there is no $\frac{123}{122}$-approximation algorithm unless P=NP.
Although the non-preemptive \ac{DaRP} appears unlikely to admit an $\O(1)$-approximation ratio, the \ac{TSP} lower bound remains the best hardness result.
For the preemptive \ac{DaRP}, \citet{DBLP:conf/approx/Gortz06} proved a super-constant lower bound of $\Omega(\min\{\log^{\frac{1}{4}-\varepsilon}\vertexnum, \lambda^{1-\varepsilon}\})$ for any constant $\varepsilon>0$ under the assumption that all problems in NP cannot be solved by randomized algorithms with expected running time $\O(\vertexnum^{\text{polylog}\  \vertexnum})$.

\textbf{\ac{mDaRP}.}  
The \ac{mDaRP} is a natural generalization of the \ac{DaRP}. 
For the non-preemptive \ac{mDaRP}, \citet{DBLP:journals/pvldb/LuoFDG23} proposed an $\O(\sqrt{\lambda}\log m)$-approximation algorithm that runs in $\O(m^3\log\lambda+m^2\lambda^2\log\lambda)$ time and performs well in practice. 
Unlike previous approximation algorithms for the \ac{DaRP}, their approach uses an \emph{iterative pairing} strategy based on minimum-weight matching, grouping requests before planning routes.

For the preemptive \ac{mDaRP}, \citet{DBLP:conf/esa/GortzNR09,DBLP:journals/talg/GortzNR15} showed that an $\O(\log\vertexnum)$-approximation algorithm can be obtained by extending the $\O(\log\vertexnum)$-approximation algorithm for the preemptive \ac{DaRP}~\citep{DBLP:conf/focs/CharikarR98,DBLP:journals/jcss/FakcharoenpholRT04}.

The \ac{mDaRP} has also been studied under the \emph{makespan} objective, which aims to minimize the maximum travel distance among all vehicles.
For the preemptive case, \citet{DBLP:conf/esa/GortzNR09,DBLP:journals/talg/GortzNR15} proposed an $\O(\log^3 n)$-approximation algorithm, and an $\O(\log n)$-approximation algorithm for the case $\lambda = \infty$.
They also showed that for special classes of metrics induced by graphs excluding any fixed minor (e.g., planar metrics), these ratios can be improved to $\O(\log^2 n)$ and $\O(1)$, respectively.
Moreover, they proposed an $\O(1)$-approximation algorithm for the case $\lambda = 1$, i.e., multi-vehicle \ac{SCP}, which even computes a non-preemptive solution.

\subsection{Our results}
For the \ac{mDaRP}, we design new approximation algorithms. 
The contributions are summarized as follows.

\begin{enumerate}
\item We first propose two simple approximation algorithms that are particularly practical when $\rnum = \Theta(\vertexnum)$.
The first achieves an approximation ratio of $\O(\sqrt{\lambda}\log\rnum)$ with a running time of $\O(\rnum^2)$, whereas the second achieves a ratio of $\O(\sqrt{\frac{\rnum}{\lambda}})$ with a running time of $\O(\rnum^2\log\rnum)$. 
Therefore, the combination of these two algorithms yields a ratio of $\O(\min\{\sqrt{\lambda}\log\rnum,\sqrt{\frac{\rnum}{\lambda}}\})$, which is at most $\O(\sqrt[4]{n}\log^{\frac{1}{2}}n)$ when $\rnum = \Theta(\vertexnum)$ and thus breaks the \emph{barrier} $\O(\sqrt{n})$ when $\rnum = o(\vertexnum^2)$.
Note that for the \ac{mDaRP}, the previous algorithm~\citep{DBLP:journals/pvldb/LuoFDG23} is relatively complex
and achieves a ratio of $\O(\sqrt{\lambda}\log\rnum)$ with a running time of $\O(\rnum^3\log\lambda+\rnum^2\lambda^2\log\lambda)$.

\item Although we may have $\rnum,\lambda\gg \vertexnum$, there are at most $\vertexnum(\vertexnum-1)$ number of distinct requests. Building upon our two previous algorithms and the techniques from~\citep{DBLP:journals/talg/GuptaHNR10}, we give our third algorithm for the \ac{mDaRP}, which achieves a ratio of $\O(\sqrt{\vertexnum\log\vertexnum})$. Notably, this result improves upon the best-known ratio of $\O(\sqrt{\vertexnum}\log^2\vertexnum)$ for the \ac{DaRP}~\citep{DBLP:conf/esa/GuptaHNR07,DBLP:journals/talg/GuptaHNR10}, which has remained unchanged for nearly two decades.
\end{enumerate}

Our algorithms are based on new structural properties that are also valuable for related \acp{DaRP}. 
Specifically, the first algorithm is built on a simple metric graph in which each request is represented by a single vertex, thereby enabling us to use techniques from the well-known \ac{CVRP} (recall that in which, each request involves a single location).
The second algorithm uses a simple yet elegant decomposition of a Steiner forest. 
Although these ideas are surprisingly simple, they yield strong theoretical guarantees through careful analysis.

\section{Preliminaries}
Let $G=(V,E,w)$ denote the input complete graph with $\size{V}=\vertexnum$ and edge-weight function $w:E\to\mathbb{R}^+$.
We call $G$ a \emph{metric graph} if $w$ satisfies the metric properties defined earlier. In this case, $G$ can also be viewed as a \emph{metric space} $(V,w)$.
There is a (multi-)set of $\rnum$ requests $\R=\{(s_i,t_i)\}_{i=1}^{\rnum}\subseteq V\times V$. Each request $\{s_i,t_i\}\in\R$ corresponds to an item that needs to be picked up from the source $s_i$ and delivered to the destination $t_i$.
We let $X=\{s_i\}_{i=1}^{\rnum}$ and $Y=\{t_i\}_{i=1}^{\rnum}$ denote the sets of all sources and destinations, respectively.
Furthermore, there is a set of $\vehiclenum$ vehicles, each with the same capacity $\lambda\in\mathbb{N}^+$.
Their initial locations are given by $K=\{o_i\}_{i=1}^{\vehiclenum}\subseteq V$.

We assume that every vertex in $V$ is either an initial location of a vehicle or the source or destination of a request; otherwise, since $G$ is a metric graph, such a vertex can be removed.
For ease of analyzing the running time of our algorithms, we also assume $\rnum = \Omega(\vertexnum)$, following~\citep{DBLP:journals/pvldb/LuoFDG23}.
Under this assumption, for example, the running time may be expressed simply as $\O(\rnum^2)$; otherwise, it would be $\O(\max\{\vertexnum^2,\rnum^2\})$, since reading the input alone already requires $\O(\vertexnum^2)$ time.
Moreover, for the \ac{mDaRP}, we assume, at the cost of losing a factor of 2, that each location in $V$ hosts at most one vehicle, i.e., $\vehiclenum \leq \vertexnum$. This is because if a vertex hosts multiple vehicles, their routes can be served by a single vehicle with a total weight at most twice the original, by the triangle inequality. Hence, for the \ac{mDaRP}, we have
\begin{equation}\label{simpleass}
\vehiclenum \leq \vertexnum = \O(\rnum).
\end{equation}

We introduce some notations on (multi-)graphs. Let $G'=(V',E', w')$ be a (multi-)graph.
In what follows, we work with multi-edge sets, where the union of two edge sets is taken with multiplicities.

For any edge set $E_s\subseteq E'$, define $w'(E_s)\coloneqq\sum_{e\in E_s}w'(e)$.  
For any subgraph $G_s$ of $G'$, let $V'(G_s)$ and $E'(G_s)$ denote its vertex set and edge set, respectively.  
For any vertex set $V_s\subseteq V'$, let $G'[V_s]$ denote the subgraph of $G'$ induced by $V_s$.

A \emph{walk} $W=v_1v_2\dots v_s$ in $G'$ is a sequence of vertices (vertices may repeat), corresponding to the edge set $E'(W)\coloneqq\{v_1v_2,v_2v_3,\dots,v_{s-1}v_s\}$.  
A \emph{tour} $T=v_1v_2\dots v_s$ is similarly a sequence of vertices, but with the edge set $E'(T)\coloneqq\{v_1v_2,v_2v_3,\dots,v_sv_1\}$.  
The weight of a walk $W$ (resp., tour $T$) is $w'(W)\coloneqq w'(E'(W))$ (resp., $w'(T)\coloneqq w'(E'(T))$).
Note that a sequence of vertices can represent either a walk or a tour, depending on whether the first and last vertices are connected; the context will always clarify which applies.
Moreover, a tour $T$ is a \emph{simple tour} if it contains no repeated vertices, and a \emph{\ac{TSP} tour} in $G'$ if additionally $\size{E'(T)}=\size{V'}$.

Given a walk or tour ${\dots xyz\dots }$ in $G'$, a new walk or tour ${\dots xz\dots }$ can be obtained by skipping the vertex $y$. This operation is called \emph{shortcutting} the vertex $y$ or the edges $xy$ and $yz$. Since the weight function $w'$ is a metric, the new walk or tour has a non-increasing weight by the triangle inequality.

For any integer $p\in\mathbb{N}^+$, let $[p]\coloneqq\{1,2,\dots,p\}$.
Consider a tour $os_{i_1}s_{i_2}\dots s_{i_p}$ visiting $p$ sources of requests in $\{(s_{i_j},t_{i_j})\mid j\in[p]\}$ and the corresponding tour $ot_{i_1}t_{i_2}\dots t_{i_p}$ visiting the $p$ destinations, where $o$ denotes an \emph{optional} vehicle.  
Such a pair of tours, which share the same permutation of indices $i_1i_2\dots i_p$, is called \emph{consistent}.  
More generally, if the index sequences for the sources and destinations are both permutations of $\{i_1,i_2,\dots,i_p\}$ (not necessarily identical), we call them \emph{weakly consistent}.  
Whenever $o$ appears in one tour of a (weakly) consistent pair, it must also appear in the other.

Similarly, let $\T_s$ and $\T_t$ be two sets of tours, where every tour in $\T_s$ consists of no destinations and every tour in $\T_t$ consists of no sources.
We call these sets (weakly) consistent if, for every tour in $\T_s$, there exists a tour in $\T_t$ that is (weakly) consistent with it, and vice versa.
In this case, we must have $\size{\T_s} = \size{\T_t}$.

A walk is \emph{simple} if it contains no repeated vertices; such a walk is called a \emph{path}.
Two paths, or two sets of paths, are said to be (weakly) consistent in the same manner as tours, with the same requirement on the presence of an optional vehicle $o$. 
{(For paths, if $o$ exists, we require that the paths start from $o$.)}
For instance, the tours $os_1s_2$ and $ot_1t_2$ are consistent;  
{the paths $s_1os_2$ and $t_1ot_2$ are not (weakly) consistent.}
For any two sequences $T_1=p_1p_2\dots p_i$ and $T_2=q_1q_2\dots q_j$, we use $T_1\circ T_2$ (resp., $T_1\bullet T_2$) to denote the sequence obtained by concatenating $T_2$ to $T_1$ (resp., removing all elements of $T_2$ from $T_1$). 
For a sequence $T=p_1p_2\dots p_i$, we let $\size{T}=i$.
For instance, if $T_1=12134$ and $T_2=1216$, then $T_1\circ  T_2 = 121341216$, $T_1\bullet T_2 = 34$, and $\size{T_1}=5$.

In the following, we denote by $\OPT$ the weight of an optimal solution to the \ac{mDaRP}.

\section{Two Simple Algorithms for the \ac{mDaRP}}
In this section, we propose two algorithms, denoted by \textsc{Alg.1} and \textsc{Alg.2}.
Recall that $\vehiclenum \leq \vertexnum$ and $\vertexnum = \O(\rnum)$ by~(\ref{simpleass}), and that $X=\{s_i\}_{i=1}^{\rnum}$, $Y=\{t_i\}_{i=1}^{\rnum}$, and $K=\{o_i\}_{i=1}^{\vehiclenum}$ by definition.  
For convenience, we assume in this section that all vertices in $X$, $Y$, and $K$ are pairwise disjoint because we can consider the same vertex as two distinct vertices with zero distance. 
Then $V = X \cup Y \cup K$ and $\vertexnum = \vehiclenum + 2\rnum = \Theta(\rnum)$.

\subsection{The First Algorithm}
Our first algorithm, denoted by \textsc{Alg.1}, is inspired by the well-known \ac{ITP} algorithm for the CVRP~\citep{HaimovichK85}.  
The \ac{ITP} works as follows: (1) compute an $\alpha$-approximate \ac{TSP} tour in $G$; (2) partition this tour into the desired path fragments; and (3) transform each fragment into a feasible tour.
Despite its simplicity, \ac{ITP} achieves a strong approximation guarantee of $\alpha+2$ for the CVRP~\citep{altinkemer1987heuristics}.  
Surprisingly, it has not been explored for approximation algorithms in the \ac{DaRP} setting.

We begin with the \ac{DaRP}, and later show how to extend the method to the \ac{mDaRP} with minor modifications.

To modify \ac{ITP} to work for the \ac{DaRP}, we note that instead of finding an approximate \ac{TSP} tour in $G$ as (1) in \ac{ITP}, we may find two approximate \ac{TSP} tours $T_s$ in $G[X\cup K]$ and $T_t$ in $G[Y\cup K]$ since each request involves two locations. 
Moreover, we may ensure that these two \ac{TSP} tours are consistent, say $T_s=os_{i_1}s_{i_2}\dots s_{i_\rnum}$ and $T_t=ot_{i_1}t_{i_2}\dots t_{i_\rnum}$. 
This is important, as for any fragment $B^i_s=s_{i_{j}}s_{i_{j+1}}\dots s_{i_{j+j'}}$ from $T_s$, there always exists a fragment $B^i_t=t_{i_{j}}t_{i_{j+1}}\dots t_{i_{j+j'}}$ from $T_t$ that is consistent with $B^i_s$.
Then, as (2) in \ac{ITP}, we construct a set of fragments $\B_s$ (resp., $\B_t$) by partitioning $T_s$ (resp., $T_t$). Note that each fragment is a path, and we may require that $\B_s$ is consistent with $\B_t$. 
Last, as (3) in \ac{ITP}, we transform each pair of consistent fragments in  $\B_s\cup\B_t$ into a tour. This can be done by doubling all edges in the pair of consistent fragments, adding two copies of the minimum-weight edge connecting these fragments, and then taking shortcuts.
Let $\T$ be the set of tours obtained. 
Traversing the \ac{TSP} tour in $T_s$ together with the tours in $\T$ yields a feasible \ac{DaRP} solution.

To extend the above idea to the \ac{mDaRP}, we seek two consistent sets of tours: $\T_s$ in $G[X\cup K]$ and $\T_t$ in $G[Y\cup K]$, such that 
(a) each source in $X$ (resp., destination in $Y$) appears in exactly one tour in $\T_s$ (resp., $\T_t$);
(b) each vehicle in $K$ appears in exactly one tour in $\T_s$ and in exactly one tour in $\T_t$;
(c) each tour in $\T_s \cup \T_t$ contains exactly one vehicle.
Given any pair of consistent tours from $\T_s\cup\T_t$, we can apply the same procedure as in the \ac{DaRP} to construct a route that serves all requests in the pair
(see \textsc{Alg.1} in Algorithm~\ref{alg1}).

\begin{algorithm}[t]
\caption{The first algorithm for \ac{mDaRP} (\textsc{Alg.1})}
\label{alg1}
\small
\begin{algorithmic}[1]
\Require An instance $G=(V,E,w)$ of the \ac{mDaRP}.
\Ensure A solution to the \ac{mDaRP}.

\State Compute two consistent sets of tours $\T_s$ and $\T_t$ satisfying conditions (a)-(c) (see Sect.~\ref{part1}).

\State Compute a set of routes $\W$ for all vehicles in $K$ (see Sect.~\ref{part2}).

\State \Return $\W$.

\end{algorithmic}
\end{algorithm}

Next, we show in Sect.~\ref{part1} how to compute the consistent sets of tours $\T_s$ and $\T_t$ satisfying conditions (a)-(c). Then, we show in Sect.~\ref{part2} how to compute the fragments and the routes of the vehicles.

\subsubsection{Compute Two Consistent Sets of Tours}\label{part1}
To compute the consistent sets of tours $\T_s$ in $G[X\cup K]$ and $\T_t$ in $G[Y\cup K]$, we employ an approximate solution to the \aca{mTSP} (\ac{mTSP})~\citep{DBLP:journals/orl/XuXR11}. 

First, we construct a new complete graph $H=(V_H,E_H,\widetilde{w})$, where $V_H=V^1_H\cup V^2_H$, ${V^1_H}=\{k_1,\dots,k_\vehiclenum\}$, and ${V^2_H}=\{l_1,\dots,l_\rnum\}$. Each vertex $k_i\in V^1_H$ refers to a vehicle $o_i\in K$, and each vertex $l_j\in V^2_H$ refers to a request $(s_j,t_j)\in\R$.
For any $x,y\in V_H$, we define the weight function $\widetilde{w}$ as follows:
{\begin{equation}\label{thefunction}
\widetilde{w}(y,x) = \widetilde{w}(x,y)\qand\widetilde{w}(x,y)=
\begin{cases}
2w(o_i,o_j), & x=k_i\in V^1_H, y=k_j\in V^1_H,\\
w(o_i,s_j)+w(o_i,t_j), & x=k_i\in V^1_H, y=l_j\in V^2_H,\\
w(s_i,s_j)+w(t_i,t_j), & x=l_i\in V^2_H, y=l_j\in V^2_H.
\end{cases}
\end{equation}}

Note that $\widetilde{w}$ remains a metric.
By regarding the vertices in $V^1_H$ as \emph{depots}, the graph $H$ forms a valid instance of the \ac{mTSP}, where the objective is to compute a set of \emph{\ac{mTSP} tours} minimizing the total weight, subject to the conditions that: (1) each vertex in $V_H$ is contained in exactly one tour, and (2) each tour contains exactly one depot in $V^1_H$. Note that a tour may consist of a single depot only.

Since $\vehiclenum\leq \vertexnum=\O(\rnum)$ by (\ref{simpleass}), we have $\size{V_H}=\O(m)$.
Therefore, for the \ac{mTSP}, a $2$-approximate solution, denoted by $\T_{st}$, can be found in $\O(\rnum^2)$ time~\citep{DBLP:journals/orl/XuXR11}. Then, according to $\T_{st}$, two consistent sets of tours $\T_s$ and $\T_t$ can be obtained as follows: initialize $\T_s\coloneqq\emptyset$ and $\T_t\coloneqq\emptyset$; then, for each tour $k_{i'}l_{i_{j}}l_{i_{j+1}}\dots l_{i_{j+j'}}\in\T_{st}$, update $\T_s\coloneqq\T_s\cup \{o_{i'}s_{i_{j}}s_{i_{j+1}}\dots s_{i_{j+j'}}\}$ and $\T_t\coloneqq\T_t\cup \{o_{i'}t_{i_{j}}t_{i_{j+1}}\dots t_{i_{j+j'}}\}$ (see details in Algorithm~\ref{algpart1}).

\begin{algorithm}[t]
\caption{Compute two consistent sets of tours}
\label{algpart1}
\small
\begin{algorithmic}[1]
\Require An instance $G=(V,E,w)$ of the \ac{mDaRP}.

\Ensure Two consistent sets of tours $\T_s$ in $G[X\cup K]$ and $\T_t$ in $G[Y\cup K]$.

\State Construct a complete graph $H=(V_H,E_H,\widetilde{w})$, where $V_H=V^1_H\cup V^2_H$, ${V^1_H}=\{k_1,\dots,k_m\}$, and ${V^2_H}=\{l_1,\dots,l_n\}$. Moreover, the weight function $\widetilde{w}$ is defined as in (\ref{thefunction}).

\State\label{2approx} Compute a $2$-approximate solution $\T_{st}$ to the \ac{mTSP} in $H$ by using the algorithm in~\citep{DBLP:journals/orl/XuXR11}, where the vertices in $V^1_H$ are regarded as depots.

\State Initialize $\T_s\coloneqq\emptyset$ and $\T_t\coloneqq\emptyset$.

\State For each tour $k_{i'}l_{i_{j}}l_{i_{j+1}}\dots l_{i_{j+j'}}\in\T_{st}$, update $\T_s\coloneqq\T_s\cup \{o_{i'}s_{i_{j}}s_{i_{j+1}}\dots s_{i_{j+j'}}\}$ and $\T_t\coloneqq\T_t\cup \{o_{i'}t_{i_{j}}t_{i_{j+1}}\dots t_{i_{j+j'}}\}$. \label{coins}

\State\Return $\T_s$ and $\T_t$.
\end{algorithmic}
\end{algorithm}

It is clear that Algorithm~\ref{algpart1} takes $\O(\rnum^2)$ time. 
We have the following properties.

\begin{lemma}\label{thetours1}
Algorithm~\ref{algpart1} computes, in $\O(\rnum^2)$ time, two consistent sets of tours, $\T_s$ and $\T_t$, that satisfy conditions (a)-(c) and ensure that $\widetilde{w}(\T_{st}) = w(\T_s) + w(\T_t)$.
\end{lemma}
\begin{proof}
Since $H$ is an instance of the \ac{mTSP} with depots given by the vertices in $V^1_H$, we know by the objective of the \ac{mTSP} that (1) each vertex in $V_H$ is contained in exactly one tour in $\T_{st}$, and (2) each tour in $\T_{st}$ contains exactly one depot in $V^1_H$.
Therefore, by line~\ref{coins}, conditions (a)-(c) are all satisfied.

Moreover, by (\ref{thefunction}), for any tour $T_{st}=k_{i'}l_{i_{j}}l_{i_{j+1}}\dots l_{i_{j+j'}}\in\T_{st}$, we can easily obtain $\widetilde{w}(T_{st})=w(T_s)+w(T_t)$, where $T_s=o_{i'}s_{i_{j}}s_{i_{j+1}}\dots s_{i_{j+j'}}\in \T_s$ and $T_t=o_{i'}t_{i_{j}}t_{i_{j+1}}\dots t_{i_{j+j'}}\in\T_t$. Thus, $\widetilde{w}(\T_{st}) = w(\T_s) + w(\T_t)$.
\end{proof}

Although the graph $H$ is simple, we show through novel analysis that the consistent sets of tours $\T_s$ and $\T_t$ in Lemma~\ref{thetours1} satisfy $w(\T_s)+w(\T_t)\leq \O(\sqrt{\lambda}\log\rnum)\cdot\OPT$.
We recall the following well-known result.

\begin{lemma}[\citet{erdos1935combinatorial}]\label{erdosres}
Every permutation on $\{1,2,..,r\}$ has an increasing or decreasing subsequence of length at least $\Ceil{\sqrt{r}}$.
\end{lemma}

Based on the above lemma, we first prove the following structural property.

\begin{lemma}\label{structure}
Given any two metric spaces $M_s=(V_s=\{s_1,\dots,s_r\}, w_s)$ and $M_t=(V_t=\{t_1,\dots,t_r\}, w_t)$, there exists a \ac{TSP} tour $T$ in $M_{st}\coloneqq (V_{st}=\{1,\dots,r\}, l)$, where $l(i,j)\coloneqq w_s(s_i,s_j)+w_t(t_i,t_j)$, such that $l(T)\leq 2\sqrt{r-1} (w_s(T^*_s)+w_t(T^*_t))$ where $T^*_s$ (resp., $T^*_t$) is an optimal \ac{TSP} tour in $M_s$ (resp., $M_t$).
\end{lemma}
\begin{proof}
Suppose $T^*_s={s_1s_2\dots s_r}$ and $T^*_t={t_{p_1}t_{p_2}\dots t_{p_r}}$, where $p_1p_2\dots p_r$ is a permutation on $\{1,\dots,r\}$. 
Since $T^*_t$ is a \ac{TSP} tour in $M_t$, we assume w.l.o.g.\ that $p_1=1$.
To prove the lemma, it suffices to prove that there exists a \ac{TSP} tour $T={q_1q_2\dots q_r}$ in $M_{st}$, where $q_1q_2\dots q_r$ is a permutation on $\{1,\dots,r\}$, such that $l(T)=$ $w_s( T_s)+w_t( T_t)\leq 2\sqrt{r-1} (w_s(T^*_s)+w_t(T^*_t))$, where $T_s={s_{q_1}s_{q_2}\dots s_{q_r}}$ and $T_t={t_{q_1}t_{q_2}\dots t_{q_r}}$.

We propose an algorithm to construct such a permutation.
The main idea is as follows. 

The \ac{TSP} tours $T^*_s$ and $T^*_t$ correspond to the permutations $1\dots r$ and $p_1\dots p_r$, respectively.
By Lemma~\ref{erdos}, the sequence $p_2\dots p_r$ contains either an increasing or decreasing subsequence of length at least $\Ceil{\sqrt{r-1}}$. 
By traveling along $T^*_s$ and $T^*_t$ once each, we can extract two consistent tours $T'_s$ and $T'_t$ corresponding to this subsequence, and by the triangle inequality, 
$
w_s(T'_s)+w_t(T'_t)\leq w_s(T^*_s)+w_t(T^*_t).
$
Then, we remove the elements of this subsequence from $p_2 \dots p_r$ and repeat the process on the remaining sequence, obtaining another pair of consistent tours. By iterating this procedure and combining all such tours, we eventually obtain two consistent \ac{TSP} tours corresponding to the same permutation.
The details are given in Algorithm~\ref{algo:conincident}.

\begin{algorithm}[t]
\caption{Construct a good permutation}
\label{algo:conincident}
\small

\begin{algorithmic}[1]
\Require Optimal \ac{TSP} tours $T^*_s=s_1s_2\dots s_r$ and $T^*_t=t_{p_1}t_{p_2}\dots t_{p_r}$ in the metric spaces $M_s$ and $M_t$, where $p_1=1$.

\Ensure A permutation $q_1q_2\dots q_r$ on $\{1,\dots,r\}$.

\State Initialize $i\coloneqq 1$, $T^s_1\coloneqq 23\dots r$, $T^t_1\coloneqq p_2p_3\dots p_r$, and $T_1=1$.

\While{$f_i\coloneqq r-\size{ T_i}\geq 1$}\label{startloop}

\State\label{erdos} Find an increasing or decreasing subsequence $\widetilde{T}_i$ of length at least $\Size{\widetilde{T}_i}\geq \Ceil{\sqrt{f_i}}$ in $T^t_i$.

\State\label{cons} Update $T^s_{i+1}\coloneqq T^s_i\bullet\widetilde{T}_i$,
$T^t_{i+1}\coloneqq T^t_i\bullet \widetilde{T}_i$,
$T_{i+1}= T_i\circ \widetilde{T}_i$, and $i\coloneqq i+1$.

\EndWhile\label{endloop}

\State \Return $T_{i^*+1}$, where $i^*$ denotes the number of iterations.

\end{algorithmic}
\end{algorithm}

Note that $f_i=r-\size{ T_i}$ and $i^*$ denotes the number of iterations in lines~\ref{startloop}-\ref{endloop}. We have $f_1=r-1$, $f_{i^*+1}=0$, and $f_i\in \mathbb{N}$ for all $i\in[i^*+1]$. 
Moreover, by line~\ref{erdos}, for each $i\in[i^*]$, we have
\begin{equation}\label{ditui}
f_{i+1}\leq f_i-\Ceil{\sqrt{f_i}}<f_i.
\end{equation}

Next, we give an upper bound on $i^*$.
For any $i\in[i^*]$, we have
$\sqrt{f_{i+1}}-\sqrt{f_{i}}\leq \frac{-\Ceil{\sqrt{f_{i}}}}{\sqrt{f_{i+1}}+\sqrt{f_{i}}}\leq \frac{-\sqrt{f_{i}}}{\sqrt{f_{i+1}}+\sqrt{f_{i}}}<-\frac{1}{2}$, where the first inequality follows from $f_{i+1}-f_i\leq -\Ceil{\sqrt{f_i}}$ by (\ref{ditui}), the second from $\Ceil{\sqrt{f_{i}}}\geq \sqrt{f_{i}}$, and the last from $f_{i+1}<f_i$ by (\ref{ditui}).
Summing over $i\in[i^*]$, we obtain
$\sqrt{f_{i^*+1}}-\sqrt{f_1}=\sum_{i=1}^{i^*}\lrA{\sqrt{f_{i+1}}-\sqrt{f_{i}}}<-\frac{i^*}{2}$.
Since $f_1=r-1$ and $f_{i^*+1}=0$, we obtain
$i^*\leq 2\sqrt{f_1}=2\sqrt{r-1}$.

Then, we give an upper bound on $l( T_{i^*+1})$.
By line~\ref{cons}, $T_{i^*+1}=1\circ \widetilde{T}_1\circ \widetilde{T}_2\circ\cdots\circ  \widetilde{T}_{i^*}$.
Let $T_s$ and $T_t$ denote the \ac{TSP} tours in $M_s$ and $M_t$ corresponding to $T_{i^*+1}$, respectively. Note that $T_s$ and $T_t$ are consistent, and we have $l(T_{i^*+1})=w_s(T_s)+w_t(T_t)$. We may only analyze $w_s(T_s)$.
For each $j\in[i^*]$, let $\size{\widetilde{T}_j}=r_j$ and $\widetilde{T}_j=\widetilde{p}_{j_1}\dots \widetilde{p}_{j_{r_j}}$. Then, we have
$T_s=s_1\dots s_{\widetilde{p}_{i_1}}\dots s_{\widetilde{p}_{i_{r_i}}}\dots s_{\widetilde{p}_{i^*_1}}\dots s_{\widetilde{p}_{i^*_{r_{^*}}}}$.
By the triangle inequality, the weight $w_s$ of the tour $T_s$ corresponding to the permutation $T_{i^*+1}$ is at most the weight $w_s$ of the tour $T'_s$ corresponding to the sequence $T'_{i^*+1}=1\circ \widetilde{T}_1\circ1\circ \widetilde{T}_2\circ\cdots\circ1\circ  \widetilde{T}_{i^*}$, i.e., $w_s( T_s)\leq w_s(T'_s)$.
Note that $T'_s$ consists of $i^*$ tours, which correspond to the sequences $1\circ \widetilde{T}_1,\dots,1\circ \widetilde{T}_{i^*}$, respectively.
Since $\widetilde{T}_i$ is a subsequence of the permutation $2\dots r$ or $r\dots 2$, the weight $w_s$ of the tour corresponding to the sequence $1\circ \widetilde{T}_i$ is at most $w_s(T^*_s)$ by the triangle inequality.
Since $T'_s$ is a concatenation of $i^*$ such tours intersecting at $s_1$, we have
$w_s( T_s)\leq w_s(T'_s)\leq i^*\cdot w_s(T^*_s)$ by the triangle inequality.
Similarly, for $w_t(T_t)$, we have $w_t( T_t)\leq i^*\cdot w_t(T^*_t)$.

Therefore, we obtain $l( T_{i^*+1})=w_s( T_s)+w_t( T_t)\leq 2\sqrt{r-1} (w_s(T^*_s)+w_t(T^*_t))$.
\end{proof}

We next recall another structural property of the \ac{DaRP}.

\begin{lemma}[\citet{DBLP:journals/talg/GuptaHNR10}]\label{structure2}
Given any instance of the \ac{DaRP}, there exists a feasible walk $W$ satisfying the following conditions: (1) $W$ can be partitioned into a set of segments $\{W_1,W_2,\dots,W_N\}$, where in each $W_i$ the vehicle serves a subset $\R_i\subseteq \R$ of at most $\lambda$ requests such that $W_i$ is a path that first picks up each; (2) the weight of $W$ is at most $\O(\log\rnum)$ times that of an optimal solution, i.e., $w(W)\leq \O(\log\rnum)\cdot\OPT$.
\end{lemma}

Lemma~\ref{structure2} guarantees the existence of a solution satisfying the \emph{batch delivery} property whose total weight is at most $\O(\log\rnum)\cdot \OPT$, where this property means the vehicle serves requests in batches, each consisting of picking up at most $\lambda$ requests first, followed by delivering all of them before proceeding to the next batch. 
As shown in~\citep{DBLP:journals/pvldb/LuoFDG23}, this structural property extends trivially to the \ac{mDaRP}. 

Next, we use Lemmas~\ref{structure} and \ref{structure2} to analyze the quality of the consistent sets of tours $\T_s$ and $\T_t$ in Algorithm~\ref{algpart1}.

\begin{lemma}\label{thetours2}
The consistent sets of tours $\T_s$ and $\T_t$ in Algorithm~\ref{algpart1} satisfy $w(\T_s)+w(\T_t)\leq \O(\sqrt{\lambda}\log\rnum)\cdot\OPT$.
\end{lemma}
\begin{proof}
Let $\W$ be a solution to the \ac{mDaRP} in $G$ with weight at most $\O(\log\rnum)\cdot \OPT$, where each walk in $\W$ satisfies the batch delivery property from Lemma~\ref{structure2}.
By shortcutting, we assume w.l.o.g.\ that (1) each request is served by exactly one walk in $\W$, (2) each vehicle appears in exactly one walk, and each walk contains exactly one vehicle in $K$.

Next, we construct from $\W$ a set of \ac{mTSP} tours $\T^W$ in the graph $H$ such that $\widetilde{w}(\T^W)\leq \O(\sqrt{\lambda}\log\rnum)\cdot\OPT$. 
Since Algorithm~\ref{algpart1} computes a 2-approximate \ac{mTSP} solution $\T_{st}$ in $H$ (see line~\ref{2approx}), and Lemma~\ref{thetours1} guarantees that $\widetilde{w}(\T_{st})=w(\T_s)+w(\T_t)$, it follows that $w(\T_s)+w(\T_t)\leq 2\widetilde{w}(\T^W)\leq \O(\sqrt{\lambda}\log\rnum)\cdot\OPT$.

For any walk $W\in\W$, it begins at a vehicle location, saying $o_W\in K$. By the batch delivery property, we can express the walk as $W=o_W\circ W_1\circ W_2\circ \dots W_N$, where in each walk $W_i$ the vehicle picks up at most $\lambda$ requests first, followed by delivering all of them. 
Therefore, each walk $W_i$ can be further partitioned into two walks $W^1_i$ and $W^2_i$, where the vehicle picks up the requests in $ W^1_i$ and delivers them in $ W^2_i$.

Then, for each walk $W_i$, we construct two simple tours $T^s_i$ and $T^t_i$ by doubling and shortcutting $W^1_i$ and $W^2_i$, respectively. By the triangle inequality, we have $w(T^s_i)+w(T^t_i)\leq 2w(W^1_i)+2w(W^2_i)\leq 2w(W_i)$.
By Lemma~\ref{structure}, there exists a pair of consistent tours $\widetilde{T}^s_i$ and $\widetilde{T}^t_i$ such that $w(\widetilde{T}^s_i)+w(\widetilde{T}^t_i)\leq \sqrt{\lambda-1}(w(T^s_i)+w(T^t_i))\leq 2\sqrt{\lambda-1}\cdot w(W_i)$, where $\widetilde{T}^s_i$ (resp., $\widetilde{T}^t_i$) consists of the sources (resp., destinations) of all served requests by $W_i$.
Therefore, there exists a pair of consistent sets of tours $\widetilde{\T}_s\coloneq\{\widetilde{T}^s_i\}_{i=1}^{N}$ and $\widetilde{\T}_t\coloneq\{\widetilde{T}^t_i\}_{i=1}^{N}$ such that 
{\begin{equation}\label{xe1}
w(\widetilde{\T}_s)+w(\widetilde{\T}_t)\leq \sum_{i=1}^{N}2\sqrt{\lambda-1}w(W_i)=2\sqrt{\lambda-1}\cdot w(W).
\end{equation}}

Now, consider the tours 
$T^W_s=o_W\circ \widetilde{T}^s_1\circ \widetilde{T}^s_2\circ\cdots \widetilde{T}^s_N$ and
$T^W_t=o_W\circ \widetilde{T}^t_1\circ \widetilde{T}^t_2\circ\cdots \widetilde{T}^t_N$.
By construction, $\widetilde{T}^s_i$ and $\widetilde{T}^t_i$ are consistent for each $i\in[N]$, and then $T^W_s$ and $T^W_t$ are consistent. Moreover, $T^W_s$ (resp., $T^W_t$) covers all sources (resp., destinations) of requests served by $W$.

Assume $\widetilde{T}^s_i=s_{i_1}\dots $ for each $i\in[N]$.
By the construction of $T^W_s$ and the triangle inequality, we know that $w(T^W_s)\leq \sum_{i=1}^{N}w(\widetilde{T}^s_i)+w(W_s)$, where $W_s$ denotes the walk $o_Ws_{1_1}s_{2_1}\dots s_{N_1}$.
Since $W_s$ can be obtained from $W$ by shortcutting so that only the vertices in $\{o_W,s_{1_1},s_{2_1},\dots,s_{N_1}\}$ remain, we have $w(W_s)\leq w(W)$ by the triangle inequality.
Then, we have
{\begin{equation}\label{xe4}
w(T^W_s)\leq \sum_{i=1}^{N}w(\widetilde{T}^s_i)+w(W_s)=w(\widetilde{\T}_s)+w(W),
\end{equation}}
where the equality follows from $\widetilde{\T}_s=\{\widetilde{T}^s_i\}_{i=1}^{N}$ by definition. 
Similarly, we have $w(T^W_t)\leq w(\widetilde{\T}_t)+w(W)$.
Then, by (\ref{xe1}) and (\ref{xe4}), we have
{\begin{equation}\label{xe6}
w(T^W_s)+w(T^W_t)\leq 2\lrA{\sqrt{\lambda-1}+1}w(W).
\end{equation}}

According to the tours $T^W_s$ and $T^W_t$ in the graph $G$, let $T^W_{st}$ denote the corresponding tour in $H$. By Lemma~\ref{thetours1}, we have $\widetilde{w}(T^W_{st})=w(T^W_s)+w(T^W_t)$. 
Then, by the definitions of $\W$, $T^W_s$, and $T^W_t$, the set $\T^W\coloneqq\{T^W_{st}\}_{W\in\W}$ forms a feasible \ac{mTSP} solution in $H$. Moreover, by (\ref{xe6}), we have
{\begin{align*}
\widetilde{w}(\T^W)&=\sum_{W\in\W}\lrA{w(T^W_s)+w(T^W_t)}\leq 2\lrA{\sqrt{\lambda-1}+1}\sum_{W\in\W}w(W)= 2\lrA{\sqrt{\lambda-1}+1}w(\W)\leq \O(\sqrt{\lambda}\log\rnum)\cdot\OPT,
\end{align*}}
where the last inequality follows from $w(\W)\leq \O(\log\rnum)\cdot\OPT$.
\end{proof}

We remark that the upper bound $2\sqrt{r-1}$ in Lemma~\ref{structure} is tight up to a constant factor.
As an example, suppose $\sqrt{r}\in\mathbb{N}^+$, and let $V_s=\bigcup_{i=1}^{\sqrt{r}}V^i_s$ and $V_t=\bigcup_{i=1}^{\sqrt{r}}V^i_t$, where $V^i_s=\{s_{i_1},s_{i_2},\dots,s_{i_{\sqrt{r}}}\}$ and $V^i_t=\{t_{1_i},t_{2_i},\dots,t_{\sqrt{r}_{i}}\}$ for any $i\in[\sqrt{r}]$.
Moreover, the metric weight functions $w_s$ and $w_t$ are defined as follows: 
\begin{itemize}
    \item $w_s(s,s')=1$ if $s\in V^i_s$ and $s'\in V^j_s$ with $i\neq j$, and $w_s(s,s')=0$ otherwise;
    \item $w_t(t,t')=1$ if $t\in V^i_t$ and $t'\in V^j_t$ with $i\neq j$, and $w_t(t,t')=0$ otherwise.
\end{itemize}

On one hand, the optimal \ac{TSP} tours $T^*_s$ in $M_s$ and $T^*_t$ in $M_t$ satisfy that $w_s(T^*_s)=w_t(T^*_t)=\sqrt{r}$. 

On the other hand, for any two distinct pairs $(i,j)$ and $(i',j')$ with $i,i',j,j'\in[\sqrt{r}]$, we have $w_s(s_{i_j},s_{i'_{j'}})+w_t(t_{i_j},t_{i'_{j'}})=1$ because we have either $w_s(s_{i_j},s_{i'_{j'}})=1$ or $w_t(t_{i_j},t_{i'_{j'}})=1$.

Therefore, for any permutation $T=p_1p_2\dots p_r$ on $\{1,\dots,r\}$, the corresponding \ac{TSP} tours $T_s$ in $M_s$ and $T_t$ in $M_t$ satisfy that $l(T)=w_s( T_s)+w_t( T_t)=r$. 
Hence, there exists a lower bound of $\frac{r}{2\sqrt{r}}=\frac{\sqrt{r}}{2}$ for the ratio $\frac{l(T)}{w_s(T^*_s)+w_t(T^*_t)}$ in Lemma~\ref{structure}, and then our ratio $2\sqrt{r-1}$ is tight up to a constant factor.

\subsubsection{The Routes of the Vehicles}\label{part2}
In this subsection, we show how to obtain the routes of the vehicles.

Now, we have a pair of consistent sets of tours $\T_s$ and $\T_t$ satisfying the properties in Lemma~\ref{thetours1}. 
We may arbitrarily select a pair of consistent tours $T_s\in \T_s$ and $T_t\in\T_t$, and show how to obtain two consistent sets of fragments $\B_s$ and $\B_t$ by possibly partitioning $T_s$ and $T_t$ simultaneously, which will then be used to construct a route (walk) for the vehicle in $T_s$ and $T_t$.

Since $T_s$ and $T_t$ are consistent, we may let $T_s=o_{i'}s_{1}s_{2}\dots s_{i}$ and $T_t=o_{i'}t_{1}t_{2}\dots t_{i}$. We consider two cases.

\textbf{Case~1: $i\leq \lambda$.} In this case, no partitioning is needed.
We can directly construct a route $W_{o_{i'}}=T_s\circ (T_t\bullet o_{i'})=o_{i'}s_{1}s_{2}\dots s_{i}t_{1}t_{2}\dots t_{i}$ for the vehicle $o_{i'}$.

\textbf{Case~2: $i>\lambda$.} 
We choose an integer $\theta$ uniformly at random from $\{1,2,\dots,\lambda\}$.
We then partition $T_s$ and $T_t$ into two consistent sets of fragments as follows: $\B_s=\{B^1_s,B^2_s,\dots,B^{N_\theta}_s\}$ and $\B_t=\{B^1_t,B^2_t,\dots,B^{N_\theta}_t\}$, where $N_\theta=\Ceil{\frac{i-\theta}{\lambda}}+1$. Similar to the \ac{ITP} method~\citep{altinkemer1990heuristics}, the fragments in $\B_s$ are given by 
{\begin{equation}\label{thefragments}
\begin{cases}
&B^1_s=o_{i'}s_{1}s_{2}\dots s_{{\theta}},\\
&B^2_s=s_{{\theta+1}}s_{{\theta+2}}\dots s_{{\theta+\lambda}},\\
&B^3_s=s_{{\theta+\lambda+1}}s_{{\theta+\lambda+2}}\dots s_{{\theta+2\lambda}},\\
&\dots\\
&B^{N_\theta}_s=s_{{\theta+\lambda\cdot (N_\theta-2)+1}}s_{{\theta+\lambda\cdot (N_\theta-2)+2}}\dots s_{i}.
\end{cases}
\end{equation}}
Note that each $B^j_t$ in $\B_t$ is defined analogously, with $B^j_s$ consistent with $B^j_t$ for all $j\in[N_\theta]$. For instance, $B^1_t=o_{i'}t_{1}t_{2}\dots t_{{\theta}}$ and $B^2_t=t_{{\theta+1}}t_{{\theta+2}}\dots t_{{\theta+\lambda}}$.

Now, we have $N_\theta$ pairs of consistent fragments in $\B_s\cup\B_t$, each covering at most $\lambda$ requests.
These are transformed into $N_\theta$ walks $\{W^1_{o_{i'}},W^2_{o_{i'}},\dots,W^{N_\theta}_{o_{i'}}\}$, where $W^1_{o_{i'}}=B^1_s\circ (B^1_t\bullet o_{i'})$ and $W^j_{o_{i'}}=B^j_s\circ B^j_t$ for $j>1$. 

Finally, we concatenate them to form a feasible route $W_{o_{i'}}$ for the vehicle $o_{i'}$, where
{\begin{equation}\label{theroute}
W_{o_{i'}}=W^1_{o_{i'}}\circ W^2_{o_{i'}}\circ\cdots \circ W^{N_\theta}_{o_{i'}}.
\end{equation}}

The details are shown in Algorithm~\ref{algpart2}. We have the following properties.

\begin{algorithm}[t]
\caption{Compute the routes of the vehicles}
\label{algpart2}
\small

\begin{algorithmic}[1]
\Require A pair of consistent sets of tours $\T_s$ and $\T_t$ satisfying the properties in Lemma~\ref{thetours1}.

\Ensure A set of routes $\W$ for all vehicles in $K$.

\State Initialize $\W\coloneqq \emptyset$.

\For{each pair of consistent tours $T_s\in \T_s$ and $T_t\in \T_t$}\label{4lpair}
\Comment{Assume that $T_s=o_{i'}s_{1}s_{2}\dots s_{i}$ and $T_t=o_{i'}t_{1}t_{2}\dots t_{i}$.}

\State \textbf{Case~1: $i\leq\lambda$.}
Obtain a route $W_{o_{i'}}=T_s\circ (T_t\bullet o_{i'})$.\label{4ltour1} 

\State \textbf{Case~2: $i>\lambda$.}
Choose $\theta$ uniformly at random from $[\lambda]$, then obtain two consistent sets of fragments: $\B_s=\{B^1_s,B^2_s,\dots,B^{N_\theta}_s\}$ and $\B_t=\{B^1_t,B^2_t,\dots,B^{N_\theta}_t\}$, where $N_\theta=\Ceil{\frac{i-\theta}{\lambda}}+1$, and form a route $W_{o_{i'}}$ as defined in (\ref{theroute}).\label{4ltour2} 

\State Update $\W\coloneqq\W\cup\{W_{o_{i'}}\}$.
\EndFor

\State \Return $\W$.

\end{algorithmic}
\end{algorithm}

\begin{lemma}\label{theroutes1}
Algorithm~\ref{algpart2} computes in $\O(\rnum)$ time a set of feasible routes $\W$, {where exactly one route is assigned for each vehicle in $K$.} 
\end{lemma}
\begin{proof}
Since $\theta\leq \lambda$, lines \ref{4ltour1}-\ref{4ltour2} together with (\ref{thefragments}) guarantee that each $W_{o_{i'}}\in \W$ is a feasible route for its corresponding vehicle $o_{i'}$. 
Moreover, by Lemma~\ref{thetours1}, exactly one route in $\W$ is assigned to each vehicle in $K$.
 
Since we only process consistent tour pairs by line~\ref{4lpair}, all routes in $\W$ can be computed using only $\T_s$. 
By our previous assumption that $\vertexnum = \vehiclenum + 2\rnum = \Theta(\rnum)$ and Lemma~\ref{thetours1}, the total number of vertices across all tours in $\T_s$ is $\O(\vehiclenum + \rnum)\subseteq\O(\rnum)$.
Thus, the running time of Algorithm~\ref{algpart2} is at most $\O(\rnum)$.
\end{proof}

\begin{lemma}\label{theroutes2}
It holds that $\EEE{w(\W)}\leq\frac{2}{\lambda}\sum_{j=1}^{m}w(s_j,t_j)+3(w(\T_s)+w(\T_t))$.
\end{lemma}
\begin{proof}
Consider a pair of consistent tours $T_s=o_{i'}s_{1}s_{2}\dots s_{i}\in \T_s$ and $T_t=o_{i'}t_{1}t_{2}\dots t_{i}\in\T_t$. Algorithm~\ref{algpart2} constructs a route $W_{o_{i'}}$ by possibly partitioning $T_s$ and $T_t$. We have the following two cases.

\textbf{Case~1: $i\leq \lambda$.} By line~\ref{4ltour1}, we have $W_{o_{i'}}=T_s\circ (T_t\bullet o_{i'})$. By the triangle inequality, $w(W_{o_{i'}})\leq w(T_s)+w(T_t)$.

\textbf{Case~2: $i>\lambda$.} By line~\ref{4ltour2} and (\ref{theroute}), we have $W_{o_{i'}}=W^1_{o_{i'}}\circ W^2_{o_{i'}}\circ\cdots \circ W^{N_\theta}_{o_{i'}}$, where $W^1_{o_{i'}}=B^1_s\circ (B^1_t\bullet o_{i'})$ and $W^j_{o_{i'}}=B^j_s\circ B^j_t$ for $j>1$. Then, unfolding this sequence, we obtain
\begin{align*}
W_{o_{i'}}&=(o_{i'}s_{1}\dots s_{{\theta}})(t_{1}\dots t_{{\theta}})(s_{{\theta+1}}\dots s_{{\theta+\lambda}})(t_{{\theta+1}}\dots t_{{\theta+\lambda}})\dots (s_{{\theta+\lambda\cdot (N_\theta-2)+1}}\dots s_{i})(t_{{\theta+\lambda\cdot (N_\theta-2)+1}}\dots t_{i}).
\end{align*}

By the triangle inequality, it is easy to obtain $w(W_{o_{i'}})\leq 2\sum_{j=1}^{N_\theta-1}w(s_{\theta+\lambda\cdot (j-1)+1},t_{\theta+\lambda\cdot (j-1)+1})+3(w(T_s)+w(T_t))$,
where the term $\sum_{j=1}^{N_\theta-1}w(s_{\theta+\lambda\cdot (j-1)+1},t_{\theta+\lambda\cdot (j-1)+1})$ is the sum of $w(s_{j'}, t_{j'})$ over all $s_{j'}$ that appears as the first vertex of a fragment (except for $o_{i'}$). Since $\theta$ is chosen uniformly at random from $\{1,\dots,\lambda\}$, each $s_{j'}$ (with $j'>1$) is selected as the first vertex of a fragment with probability $\frac{1}{\lambda}$. Therefore, we have
{\begin{equation}\label{4eqx}
\begin{split}
\EEE{w(W_{o_{i'}})}&\leq 3(w(T_s)+w(T_t))+\frac{2}{\lambda}\sum_{j=2}^{i}w(s_j,t_j)\leq 3(w(T_s)+w(T_t))+\frac{2}{\lambda}\sum_{j=1}^{i}w(s_j,t_j).
\end{split}
\end{equation}}

By Lemma~\ref{thetours1}, the tours in $\T_s$ (resp., $\T_t$) are vertex disjoint and together cover all request sources (resp., destinations).
Therefore, by (\ref{4eqx}), we have $\EEE{w(\W)}\leq3(w(\T_s)+w(\T_t))+\frac{2}{\lambda}\sum_{j=1}^{m}w(s_j,t_j)$.
\end{proof}

Note that Algorithm~\ref{algpart2} can be derandomized in $\O(\rnum^2)$ time either by enumerating all possible values of $\lambda$ and selecting the best one, or by computing the optimal partition via a dynamic programming approach.

\begin{lemma}[Flow lower bound~\citep{DBLP:journals/talg/GuptaHNR10}]\label{thelb}
It holds that $\OPT\geq \mathrm{flowLB}\coloneq\sum_{i=1}^{\rnum}\frac{w(s_i,t_i)}{\lambda}$.
\end{lemma}

By Lemmas~\ref{thetours1}, \ref{thetours2}, \ref{theroutes1}, \ref{theroutes2}, and \ref{thelb}, we have the following approximation guarantee.

\begin{theorem}\label{mainres1}
For the \ac{mDaRP}, \textsc{Alg.1} is an $\O(\sqrt{\lambda}\log\rnum)$-approximation algorithm with runtime $\O(\rnum^2)$.
\end{theorem}

Interestingly, we have the following corollary.

\begin{corollary}
\textsc{Alg.1} achieves an approximation ratio of $\O(\sqrt{\lambda}\log\rnum)$ even when all vehicles have capacity $\Theta(\sqrt{\lambda})$.
\end{corollary}
\begin{proof}
If all vehicles have capacity $\Theta(\sqrt{\lambda})$, then the first term $\frac{2}{\lambda} \sum_{j=1}^\rnum w(s_j,t_j)$ in the upper bound of Lemma~\ref{theroutes2} scales as $\frac{1}{\Theta(\sqrt{\lambda})}\sum_{j=1}^{\rnum}w(s_j,t_j)$, since each fragment covers $\Theta(\sqrt{\lambda})$ requests.
Hence, by Lemmas~\ref{thetours2} and \ref{thelb}, the approximation ratio of \textsc{Alg.1} remains $\O(\sqrt{\lambda}\log\rnum)$.
\end{proof}

Note that the previous approximation algorithms for the \ac{DaRP} in \citep{DBLP:conf/focs/CharikarR98,DBLP:journals/talg/GuptaHNR10} also possess this property, although their approaches differ from that of \textsc{Alg.1}.

We believe the structural properties used in our \textsc{Alg.1} can be extended more related \acp{DaRP}. 
For instance, it can be easily verified that one can obtain an $\O(\sqrt{\lambda}\log\rnum)$-approximation algorithm for the non-preemptive makespan \ac{mDaRP} by applying an $\O(1)$-approximate solution to the makespan \ac{mTSP}~\citep{DBLP:journals/orl/EvenGKRS04} instead in the graph $H$ to obtain two consistent sets of tours.
Notably, we need to slightly modify the weight function $\widetilde{w}$ in \eqref{thefunction} so that $\widetilde{w}(x,y)=w(o_i,s_j)+w(o_i,t_j)+\frac{1}{\lambda}w(s_j,t_j)$ for any $x=k_i\in V^1_H$ and $y=l_j\in V^2_H$, and $\widetilde{w}(x,y)=
w(s_i,s_j)+w(t_i,t_j)+\frac{1}{\lambda}(w(s_i,t_i)+w(s_j,t_j))$ for any $x=l_i\in V^2_H$ and $y=l_j\in V^2_H$.

\subsection{The Second Algorithm}
Our second algorithm, denoted by \textsc{Alg.2}, is motivated by the following observation. When $\lambda=\Theta(\rnum)$, \textsc{Alg.1} attains its worst-case approximation ratio of $\O(\sqrt{\rnum}\log\rnum)$. In this regime, however, an $\O(1)$-approximation for the \ac{DaRP} can be obtained through a simple approach.

We begin by computing an $\O(1)$-approximate \ac{TSP} tour $T_s$ (resp., $T_t$) in $G[X\cup K]$ (resp., $G[Y\cup K]$). 
Clearly, $w(T_s)+w(T_t)=\O(1)\cdot\OPT$. Traversing the concatenated tour $T_s\circ T_t$ once allows the vehicle to serve $\lambda$ requests. Since $\lambda = \Theta(\rnum)$, all requests can be served by traversing $T_s \circ T_t$ at most $\Ceil{\frac{\rnum}{\lambda}} = \O(1)$ times.
Hence, the weight of the resulting route is at most $\O(1)\cdot (w(T_s)+w(T_t))=\O(1)\cdot\OPT$.
Therefore, when $\lambda$ is sufficiently large, there exists an approximation algorithm for the \ac{DaRP} that outperforms \textsc{Alg.1}. However, the above method may yield at best an $\O(\frac{\rnum}{\lambda})$-approximation in general. In \textsc{Alg.2}, we employ a more refined approach to achieve an improved approximation ratio of $\O(\sqrt{\frac{\rnum}{\lambda}})$ for the \ac{mDaRP}.

In \textsc{Alg.2}, we begin by grouping requests via a \emph{Steiner forest} in $G[X\cup Y]$ that covers all vertices in $X\cup Y$, ensuring that the source and destination of each request lie within the same tree. 
The task of finding such a minimum-weight forest is the classical \emph{Steiner forest problem}. 
Using the well-known primal-dual algorithm of~\citep{GoemansW95}, we obtain a $2$-approximate Steiner forest $\F$. 

We then transform $\F$ into a set of walks $\W$ with the following properties:
(a) each walk contains exactly one vehicle;
(b) each vehicle appears in exactly one walk (a walk may contain a single vehicle only);
(c) for every request, both its source and destination appear in the same walk.

To achieve this, we double all edges in $\F$, add an $\O(1)$-approximate \ac{mTSP} solution $\T$ in $G[X \cup K]$ with the depots given by $K$, which together forms a set of connected components, and then, for each component, perform shortcutting to obtain a simple walk. 

Select an arbitrary walk $W \in \W$ and let $o_W$ denote its vehicle. Let $N_W$ be the number of requests contained in $W$.
Define $W_s$ (resp., $W_t$) as the walk obtained from $W$ by shortcutting all destinations (resp., sources) and the vehicle $o_W$. 
For instance, if $W=os_1t_4s_7s_4t_1t_7$, then we have $N_W=3$, $W_s=s_1s_7s_4$ and $W_t=t_4t_7t_1$. 

As in \textsc{Alg.1}, we partition $W_s$ and $W_t$ simultaneously to form two sets of fragments $\B_s$ and $\B_t$, where each fragment in $\B_s$ (resp., $\B_t$) contains at most $\lambda$ sources (resp., destinations), but we use a new \emph{decomposition} method. 
These fragments are then combined to form a feasible route for $o_W$ serving all $N_W$ requests in $W$.

We first describe how to use the decomposition method to construct $\B_s$ and $\B_t$.
We consider two cases.

\textbf{Case~1: $N_W\leq \lambda$.} In this case, since $W_s$ (resp., $W_t$) contains at most $\lambda$ sources (resp., destinations), we simply set $\B_s=\{W_s\}$ and $\B_t=\{W_t\}$.

\textbf{Case~2: $N_W>\lambda$.} 
We may assume that $N_W$ is divisible by $\lambda$ and that $\frac{N_W}{\lambda}$ is the square of some integer. 
The reason is as follows. 
If $N_W$ does not satisfy these two properties, we can express it as $N_W=\lambda\cdot N_0^2+N'$ for some integers $N_0$ and $N'$ such that $\lambda\cdot N_0^2< N_W<\lambda\cdot (N_0+1)^2$. We then add $\lambda\cdot (N_0+1)^2-N_W$ \emph{dummy} requests to obtain a new walk $\widetilde{W}$, where the source and destination of each dummy request coincide with the vehicle's location in $W$.
The resulting walk $\widetilde{W}$ has $N_{\widetilde{W}}=\lambda\cdot (N_0+1)^2$ and thus meets the two desired properties.
Since dummy requests can be served at zero weight, the new instance is equivalent to the original.
Moreover, since $\lambda\cdot (N_0+1)^2-N_W\leq \lambda\cdot (2N_0+1)<3\lambda\cdot N_0^2\leq 3N_W$, the number of added requests is at most $3N_W$, which does significantly affect the instance size.

We then partition $W_s$ and $W_t$.
Compared with the partition procedure in \textsc{Alg.1}, the main differences are as follows: (1) in \textsc{Alg.2}, $W_s$ may not be consistent with $W_t$, and thus $\B_s$ may not be consistent with $\B_t$; (2) due to the two properties of $W$, each fragment in $\B_s$ (resp., $\B_t$) contains exactly $\lambda$ sources (resp., destinations).

Let $W_s=s_{p_1}s_{p_2}\dots s_{p_{\lambda\cdot N_0^2}}$. We first partition $W_s$ into a set of $N_0$ fragments $\B'_s=\{B^1_s,\dots,B^{N_0}_s\}$, where, for each $i\in[N_0]$, 
$B^i_s=s_{p_{(i-1)\lambda\cdot N_0+1}}s_{p_{(i-1)\lambda\cdot N_0+2}}\dots s_{p_{i\lambda\cdot N_0}}$.
For each $B^i_s\in \B'_s$, we construct a fragment $B^i_t$ by shortcutting $W_t$ so that $B^i_t$ is weakly consistent with $B^i_s$.
Hence, we also obtain a corresponding set of fragments $\B'_t=\{B^1_t,\dots,B^{N_0}_t\}$.
By construction, $\B'_s$ and $\B'_t$ are weakly consistent, and each fragment in $\B'_s$ (resp., $\B'_t$) consists of exactly $\lambda\cdot N_0$ sources (resp., destinations).

Let $B^i_t=t_{q_{(i-1)\lambda\cdot N_0+1}}t_{q_{(i-1)\lambda\cdot N_0+2}}\dots t_{q_{i\lambda\cdot N_0}}$ for each $i\in[N_0]$.
Next, we turn to partition each $B^{i}_t\in \B'_t$ into a set of $N_0$ fragments $\B^i_t=\{B^{i,1}_t,B^{i,2}_t,\dots,B^{i,N_0}_t\}$, where 
$B^{i,j}_t=t_{q_{(i-1)\lambda\cdot N_0+(j-1)\cdot\lambda+1}}t_{q_{(i-1)\lambda\cdot N_0+(j-1)\cdot\lambda+2}}\dots t_{q_{(i-1)\lambda\cdot N_0+j\cdot\lambda}}$.
Similarly, for each $B^{i,j}_t\in\B^{i}_t$, we obtain a weakly consistent fragment $B^{i,j}_s$ by shortcutting the fragment $B^i_s$.
Then, we also obtain a set of fragments $\B^i_s=\{B^{i,1}_s,B^{i,2}_s,\dots,B^{i,N_0}_s\}$, which is weakly consistent with $\B^i_t$.

Finally, let $\B_s\coloneqq\bigcup_{i\in[N_0]}\B^{i}_s$ and $\B_t\coloneqq\bigcup_{i\in[N_0]}\B^{i}_t$.
By construction, $\B_s$ and $\B_t$ are weakly consistent, and each fragment in $\B_s$ (resp., $\B_t$) contains exactly $\lambda$ sources (resp., destinations).
We may relabel $\B_s=\{B'^1_s,\dots,B'^{N_0^2}_s\}$ and $\B_t=\{B'^1_t,\dots,B'^{N_0^2}_t\}$ so that, for any $i<j$, the first vertex of $B'^i_s$ appears before that of $B'^j_s$ in $W$.

We have $N_0^2$ pairs of weakly consistent fragments in $\B_s\cup\B_t$, each serving $\lambda$ requests.
Then, we transform them into $N_0^2$ walks $\{W^1_{o_W},W^2_{o_W},\dots,W^{N_0^2}_{o_W}\}$, where $W^j_{o_W}=B'^j_s\circ B'^j_t$ for $j\in[N_0^2]$. 
Finally, we obtain a route $W_{o_W}$ for the vehicle $o_W$ by concatenating these walk: $W_{o_W}=o_W\circ W^1_{o_W}\circ W^2_{o_W}\circ\cdots \circ W^{N_0^2}_{o_W}$ (see Algorithm~\ref{alg2}).

\begin{algorithm}[t]
\caption{The second algorithm (\textsc{Alg.2}) for the \ac{mDaRP}}
\label{alg2}
\small
\begin{algorithmic}[1]
\Require An instance $G=(V,E,w)$ of the \ac{mDaRP}.

\Ensure A solution to the \ac{mDaRP}.

\State Initialize $\W_o\coloneqq \emptyset$.

\State\label{alg5s2} Compute a $2$-approximate Steiner forest $\F$ in $G[X\cup Y]$ using the algorithm in~\citep{GoemansW95}. 

\State\label{alg5s3} Compute a $2$-approximate \ac{mTSP} solution $\T$ in $G[X\cup K]$ with depots $K$ using the algorithm in~\citep{DBLP:journals/orl/XuXR11}. 

\State\label{alg5s4} Obtain a set of walks $\W$ by doubling all edges in $\F$, adding the edges of $\T$, and then shortcutting each component.

\For{each walk $W\in \W$}\label{alg5s5}

\State\label{alg5s7} Let $o_W$ and $N_w$ denote the vehicle and the number of requests in $W$. 
Moreover, let $W_s$ (resp., $W_t$) denote the walk obtained by shortcutting all destinations (resp., sources) and the vehicle $o_W$ from $W$.

\If{$N_W\leq\lambda$}\label{alg5s8}

\State\label{alg5s9} Let $\B_s=\{W_s\}$ and $\B_t=\{W_t\}$.

\Else\label{alg5s10}

\State\label{alg5s11} Let $N_W=\lambda\cdot N_0^2$ by possibly adding some dummy requests at $o_W$, and let $W_s=s_{p_1}s_{p_2}\dots s_{p_{\lambda\cdot N_0^2}}$. 

\State\label{alg5s13} Obtain two sets of fragments $\B'_s=\{B^1_s,\dots,B^{N_0}_s\}$ and $\B'_t=\{B^1_t,\dots,B^{N_0}_t\}$, where, for each $i\in[N_0]$, $B^i_s=s_{p_{(i-1)\lambda\cdot N_0+1}}\dots s_{p_{i\lambda\cdot N_0}}$, and $B^i_t$ is weakly consistent with $B^i_s$, obtained by shortcutting $W_t$. 

\State\label{alg5s14} Let $B^i_t=t_{q_{(i-1)\lambda\cdot N_0+1}}t_{q_{(i-1)\lambda\cdot N_0+2}}\dots t_{q_{i\lambda\cdot N_0}}$ for each $i\in[N_0]$.

\State\label{alg5s15} For each $B^{i}_t\in \B'_t$, obtain two sets of fragments $\B^i_t=\{B^{i,1}_t,\dots,B^{i,N_0}_t\}$ and $\B^i_s=\{B^{i,1}_s,\dots,B^{i,N_0}_s\}$, where, for each $j\in[N_0]$, $B^{i,j}_s$ is weakly consistent with $B^{i,j}_t$, obtained by shortcutting $B^i_s$.

\State\label{alg5s16} Let $\B_s\coloneqq\bigcup_{i\in[N_0]}\B^{i}_s$ and $\B_t\coloneqq\bigcup_{i\in[N_0]}\B^{i}_t$.

\EndIf\label{alg5s17}

\State\label{alg5s18} Relabel $\B_s=\{B'^1_s,\dots,B'^{N_0^2}_s\}$ and $\B_t=\{B'^1_t,\dots,B'^{N_0^2}_t\}$ such that, for any $i<j$, the first vertex of $B'^i_s$ appears before that of $B'^j_s$ in $W$.

\State\label{alg5s19} Let $\W_o\coloneqq\W_o\cup\{W_{o_W}\}$, where $W_{o_W}=o_W\circ W^1_{o_W}\circ W^2_{o_W}\circ\cdots \circ W^{N_0^2}_{o_W}$ and $W^j_{o_W}=B'^j_s\circ B'^j_t$ for $j\in[N_0^2]$.
\EndFor\label{alg5s21}

\State\label{alg5s22} \Return $\W_o$.

\end{algorithmic}
\end{algorithm}

\begin{lemma}\label{theweakly}
Algorithm~\ref{alg2} computes, in $\O(\rnum^2\log\rnum)$ time, a set of walks $\W$ satisfying conditions (a)-(c), and a set of feasible routes $\W_o$ such that $w(\W_o)\leq \lrA{8\sqrt{\frac{\rnum}{\lambda}}+1}(2w(\F)+w(\T))+2\sum_{i=1}^{\rnum}\frac{w(s_i,t_i)}{\lambda}$.
\end{lemma}
\begin{proof}
By line~\ref{alg5s4}, all components are obtained by doubling all edges of the Steiner forest $\F$ and then adding the edges of the \ac{mTSP} solution $\T$. 
Since $G[X\cup K]$ is an instance of the \ac{mTSP} with depots given by the vertices in $K$, the \ac{mTSP} and Steiner forest properties ensure that (1) each component contains exactly one depot in $K$, (2) each vehicle appears in exactly one component, and (3) the source and destination of each request lie within in the same component. 
Since $\W$ is obtained by shortcutting each component, conditions (a)-(c) follow immediately. Moreover, by the triangle inequality, we have
\begin{equation}\label{eqx0}
w(\W)\leq 2w(\F) + w(\T).
\end{equation}

By lines~\ref{alg5s5}-\ref{alg5s21}, the route set $\W_o$ is obtained by constructing a route $W_{o_W}$ for the vehicle $o_w$ from each walk $W=o_W\dots\in \W$, which is feasible by construction.

We now analyze the weight of $\W_o$.
Fix an arbitrary walk $W=o_W\dots\in \W$ and consider its corresponding route $W_{o_W}$ for the vehicle $o_w$.
By line~\ref{alg5s19}, the route $W_{o_W}$ is obtained by using the weakly consistent sets of fragments $\B^s$ and $\B^t$.
If the number of requests in $W$ satisfies $N_W\leq \lambda$, then lines~\ref{alg5s8}-\ref{alg5s9} yield $w(\B_s)\leq w(W_s)$ and $w(\B_t)\leq w(W_t)$.
Otherwise, lines~\ref{alg5s10}-\ref{alg5s17} imply
\begin{equation}\label{eqx1}
w(\B_s) = \sum_{i=1}^{N_0} w(\B^i_s)\quad\quad\text{ and }\quad\quad w(\B_t) = \sum_{i=1}^{N_0} w(\B^i_t).
\end{equation}

By lines~\ref{alg5s14}-\ref{alg5s15}, for each $i\in[N_0]$, $\B^i_t$ is obtained by partitioning $B^i_t\in\B'_t$, and $\B^i_s$ is obtained by shortcutting $B^i_s$ to create a weakly consistent fragment for each fragment in $\B^i_t$. Hence, for each $i\in[N_0]$, we have 
\begin{equation}\label{eqx2}
w(\B^i_t)\leq w(B^i_t)\quad\quad\text{ and }\quad\quad w(\B^i_s)\leq N_0\cdot w(B^i_s).
\end{equation}

Similarly, by line~\ref{alg5s13}, we have
\begin{equation}\label{eqx3}
w(\B'_s)=\sum_{i=1}^{N_0}w(B^i_s)\leq w(W_s)\quad\quad\text{ and }\quad\quad w(\B'_t)=\sum_{i=1}^{N_0}w(B^i_t)\leq N_0\cdot w(W_t).
\end{equation}

Since the number of dummy requests is at most $3\rnum$, we have $N_0^2=\frac{N_W}{\lambda} \leq \frac{4\rnum}{\lambda}$. Then, by \eqref{eqx1}-\eqref{eqx3}, we have
\begin{equation}\label{eqx4}
w(\B_s)\leq N_0\cdot w(W_s) = 2\sqrt{\frac{\rnum}{\lambda}}\cdot w(W_s)
\quad\quad\text{ and }\quad\quad 
w(\B_t)\leq N_0\cdot w(W_t) = 2\sqrt{\frac{\rnum}{\lambda}}\cdot w(W_t).
\end{equation}

By line~\ref{alg5s7}, $W_s$ and $W_t$ are obtained by shortcutting $W$. Then, by the triangle inequality and \eqref{eqx4}, we have 
\begin{equation}\label{eqx5}
w(\B_s) + w(\B_t)\leq 4\sqrt{\frac{\rnum}{\lambda}}\cdot w(W).
\end{equation}
We remark that \eqref{eqx5} also holds when $N_W\leq \lambda$.

Recall from line~\ref{alg5s18} that $\B_s=\{B'^1_s,\dots,B'^{N_0^2}_s\}$ and $\B_t=\{B'^1_t,\dots,B'^{N_0^2}_t\}$.
For each $i\in[N_0^2]$, let 
\[
B'^i_s=s_{p_{(i-1)\lambda+1}}\dots s_{p_{i\lambda}} \quad\quad\text{ and }\quad\quad B'^i_t=t_{q_{(i-1)\lambda+1}}\dots t_{q_{i\lambda}},
\]
where $p_{(i-1)\lambda+1}p_{(i-1)\lambda+2}\dots p_{i\lambda}$ and $q_{(i-1)\lambda+1}q_{(i-1)\lambda+2}\dots q_{i\lambda}$ are permutations of the same set.
By line~\ref{alg5s19}, 
{\small\begin{align}
w(W_{o_W}) &\leq w(o_W, s_{p_1}) + \sum_{i=1}^{N_0^2}w(\B'^i_s) + \sum_{i=1}^{N_0^2}w(s_{p_{i\lambda}}, t_{q_{(i-1)\lambda+1}}) + \sum_{i=1}^{N_0^2}w(\B'^i_t) + \sum_{i=1}^{N_0^2-1}w(t_{q_{i\lambda}}, s_{p_{i\lambda+1}})\label{eqx6}\\
&\leq w(\B_s)+w(\B_t) + \sum_{i=1}^{N_0^2}\lrA{w(s_{p_{i\lambda}}, t_{q_{(i-1)\lambda+1}})+w(s_{p_{(i-1)\lambda+1}},t_{q_{i\lambda}})} + w(o_W, s_{p_1}) + \sum_{i=1}^{N_0^2-1}w(s_{p_{(i-1)\lambda+1}}, s_{p_{i\lambda+1}}),\notag
\end{align}}\!
where the second inequality follows from $\sum_{i=1}^{N_0^2}w(\B'^i_s) + \sum_{i=1}^{N_0^2}w(\B'^i_t) = w(\B_s)+w(\B_t)$ and $w(t_{q_{i\lambda}}, s_{p_{i\lambda+1}})\leq w(s_{p_{(i-1)\lambda+1}},t_{q_{i\lambda}})+w(s_{p_{(i-1)\lambda+1}}, s_{p_{i\lambda+1}})$ by the triangle inequality.

By line~\ref{alg5s18}, for any $i<j$, the first vertex of $B'^i_s$ appears before that of $B'^j_s$ in $W$. Thus, we have $w(o_W, s_{p_1}) + \sum_{i=1}^{N_0^2-1}w(s_{p_{(i-1)\lambda+1}}, s_{p_{i\lambda+1}})\leq w(W)$.
Moreover, for any $i\in[N_0^2]$, let $e_i$ denote the minimum-weight edge connecting a vertex of $B'^i_s$ with one of $B'^t_i$. 
By the triangle inequality, $w(s_{p_{i\lambda}}, t_{q_{(i-1)\lambda+1}})+w(s_{p_{(i-1)\lambda+1}},t_{q_{i\lambda}})\leq w(B'^i_s) + w(B'^i_t) + 2w(e_i)$. Since $w(e_i)\leq \sum_{j=1}^{\lambda}\frac{w\lrA{s_{p_{(i-1)\lambda+j}}, t_{p_{(i-1)\lambda+j}}}}{\lambda}$, by \eqref{eqx5} and \eqref{eqx6}, we obtain 
{\small\begin{align*}
w(W_{o_W}) &\leq w(\B_s)+w(\B_t) + w(W) + \sum_{i=1}^{N_0^2}\lrA{w(B'^i_s) + w(B'^i_t) + 2\sum_{j=1}^{\lambda}\frac{w\lrA{s_{p_{(i-1)\lambda+j}}, t_{p_{(i-1)\lambda+j}}}}{\lambda}}\\
& = 2(w(\B_s)+w(\B_t)) + w(W) + 2\sum_{i=1}^{N_0^2}\frac{w(s_{p_{i}},t_{p_i})}{\lambda}\leq \lrA{8\sqrt{\frac{\rnum}{\lambda}}+1}w(W)+2\sum_{i=1}^{N_0^2}\frac{w(s_{p_i},t_{p_i})}{\lambda}.
\end{align*}}

Summing over all $W\in\W$ yields $w(\W_o)\leq \lrA{8\sqrt{\frac{\rnum}{\lambda}}+1}w(\W)+2\sum_{i=1}^{\rnum}\frac{w(s_i,t_i)}{\lambda}\leq \lrA{8\sqrt{\frac{\rnum}{\lambda}}+1}(2w(\F)+w(\T))+2\sum_{i=1}^{\rnum}\frac{w(s_i,t_i)}{\lambda}$, where the last inequality follows from \eqref{eqx0}.
\end{proof}

By the triangle inequality, the Steiner forest $\F$ in line~\ref{alg5s2} satisfies $w(\F)\leq 2\cdot\OPT$, and the \ac{mTSP} solution $\T$ in line~\ref{alg5s3} satisfies $w(\T)\leq 2\cdot\OPT$. 
These bounds are also referred to as the \emph{Steiner lower bounds}~\citep{DBLP:journals/talg/GuptaHNR10}.
Combining these bounds with Lemmas~\ref{thelb} and~\ref{theweakly} yields $w(\W_o)\leq \lrA{8\sqrt{\frac{\rnum}{\lambda}}+1}\cdot 6\cdot\OPT + 2\cdot\OPT=8\lrA{6\sqrt{\frac{\rnum}{\lambda}} + 1}\cdot\OPT$.
Recall that for the \ac{mDaRP}, we assume, at the cost of losing a factor of 2, that each location in $V$ contains at most one vehicle. Therefore, the final approximation ratio is $16\lrA{6\sqrt{\frac{\rnum}{\lambda}} + 1}=\O(\sqrt{\frac{\rnum}{\lambda}})$.

\begin{theorem}\label{mainres2}
For the \ac{mDaRP}, \textsc{Alg.2} is an $\O(\sqrt{\frac{\rnum}{\lambda}})$-approximation algorithm with runtime $\O(\rnum^2\log\rnum)$.
\end{theorem}

Since $\min\lrC{\sqrt{\lambda}\log\rnum,\sqrt{\frac{\rnum}{\lambda}}} \leq \sqrt[4]{\rnum}\log^{\frac{1}{2}}\rnum$, by Theorems~\ref{mainres1} and \ref{mainres2}, we obtain the following result.

\begin{theorem}\label{mainres3}
For the \ac{mDaRP}, there exists an $\O\!\lrA{\sqrt[4]{\rnum}\log^{\frac{1}{2}}\rnum}$-approximation algorithm with runtime $\O(\rnum^2\log\rnum)$. In particular, when $\rnum = o(\vertexnum^2\log^6\vertexnum)$, the approximation ratio improves $o(\sqrt{\vertexnum}\log^2\vertexnum)$.
\end{theorem}

Theorem~\ref{mainres3} implies that when $\rnum = o(\vertexnum^2)$, our approximation ratio achieves $o\lra{\sqrt{\vertexnum}}$, thereby breaking the $\O(\sqrt{\vertexnum})$-approximation barrier.
Moreover, when $\rnum = o(\vertexnum^2\log^6\vertexnum)$, our approximation ratio even improves the previously best-known approximation ratio of $\O(\sqrt{\vertexnum}\log^2\vertexnum)$ for the \ac{DaRP}~\citep{DBLP:journals/talg/GuptaHNR10}.

We remark that the lower bounds used for analyzing the approximation ratios of our algorithms in Theorems~\ref{mainres1}-\ref{mainres3} are even lower bounds of the preemptive version of the \ac{mDaRP}.
\citet{DBLP:conf/focs/CharikarR98} showed that, for the \ac{DaRP}, the optimal value of the preemptive version can be at least $\Omega(\vertexnum^{\frac{1}{3}})$ times that of the non-preemptive version, while \citet{DBLP:journals/talg/GuptaHNR10} showed that this gap persists as $\Omega\lrA{\frac{\vertexnum^{\frac{1}{8}}}{\log^3\vertexnum}}$ even in the Euclidean plane.

\section{An Improved Approximation Algorithm for the \ac{mDaRP}}
In this section, we apply our \textsc{Alg.1} and \textsc{Alg.2} to design an $\O(\sqrt{\vertexnum\log\vertexnum})$-approximation algorithm for the \ac{mDaRP} by extending the methods from \citep{DBLP:journals/talg/GuptaHNR10}.\footnote{\citet{DBLP:journals/talg/GuptaHNR10} showed that any $\rho$-approximation algorithm for the \ac{DaRP} with $\rnum=\O(\vertexnum^4)$ can be transformed into an $\O(\rho)$-approximation algorithm for the \ac{DaRP} with arbitrary $\rnum$.
Hence, directly applying \textsc{Alg.1} and \textsc{Alg.2} to their reduction would yield only an $\O(\vertexnum\sqrt{\log\vertexnum})$-approximation algorithm for the \ac{DaRP}.}
By Theorem~\ref{mainres3}, we may assume that $\rnum=\Omega(\vertexnum^2)$.

In the graph $G=(V,E,w)$, there are $\vertexnum(\vertexnum-1)$ possible distinct requests, denoted by $\P=\{(s^*_i,t^*_i)\}_{i=1}^{\vertexnum(\vertexnum-1)}$. 
Recall that $\R$ consists of $\rnum$ requests. 
For each request $(s^*_i,t^*_i)\in\P$, let $N_i$ denote its multiplicity in $\R$, which we also refer to as its \emph{size}. 
Clearly, $\sum_{i=1}^{\vertexnum(\vertexnum-1)}N_i=\rnum$.

Let $\alpha = \Theta\lrA{\frac{\vertexnum}{\log\vertexnum}}$ and $\beta = \Theta\lrA{\sqrt{\vertexnum\log\vertexnum}}$.
We classify each request $(s^*_i,t^*_i)\in\P$ as follows: it is \emph{small} if $N_i<\frac{\lambda}{\alpha}$, \emph{big} if $N_i>\frac{\lambda}{\beta}$, and \emph{normal} otherwise.
This classification partitions $\P$ into three disjoint subsets: $\P_{small}$, $\P_{normal}$, and $\P_{big}$, corresponding to small, normal, and big requests, respectively.

We next handle the requests in $\P_{small}$, $\P_{big}$, and $\P_{normal}$, respectively.

\textbf{Small requests.} We directly apply \textsc{Alg.2} to serve all requests in $\P_{small}$. We may require that all vehicles return to their initial locations at the end, which increases the total weight by at most a factor of 2 by the triangle inequality. By Theorem~\ref{mainres2}, the resulting vehicle routes have a total weight at most
{\begin{equation}\label{part1small}
\O\!\lrA{\sqrt{\frac{\size{\P_{small}}\cdot \frac{\lambda}{\alpha}}{\lambda}}}\cdot\OPT\leq \O\!\lrA{\frac{\vertexnum}{\sqrt{\alpha}}}\cdot\OPT = \O\!\lrA{\sqrt{\vertexnum\log\vertexnum}}\cdot\OPT,
\end{equation}}
where the inequality follows from $\size{P_{small}}\leq\size{\P}=\vertexnum(\vertexnum-1)$.

\textbf{Big requests.} As in line~\ref{alg5s3} of \textsc{Alg.2}, we first compute a $2$-approximate \ac{mTSP} solution $\T$ in $G[X\cup K]$ with depots $K$ using the algorithm in~\citep{DBLP:journals/orl/XuXR11}. Note that we treat $X$ and $K$ as sets, and thus $\size{X},\size{K}\leq \vertexnum$.  
Each vehicle in $K$ then follows its assigned tour in $\T$. Whenever it visits the source of a request in $(s^*_i,t^*_i)\in\P_{big}$, it serves these requests greedily: the vehicle picks up $\lambda$ items at $s^*_i$, delivers them to $t^*_i$, returns to $s^*_i$, and repeats this process exactly $\Ceil{\frac{N_i}{\lambda}}$ times. The total route weight is at most $w(\T) + \sum_{(s^*_i,t^*_i)\in\P_{big}}\Ceil{\frac{N_i}{\lambda}}\cdot 2w(s^*_i, t^*_i)$. Note that $w(\T)\leq \O(1)\cdot\OPT$ and 
\[
\sum_{(s^*_i,t^*_i)\in\P_{big}}\Ceil{\frac{N_i}{\lambda}}\cdot 2w(s^*_i, t^*_i)\leq \sum_{(s^*_i,t^*_i)\in\P_{big}}\lrA{1+\beta}\frac{2N_i}{\lambda}w(s^*_i, t^*_i)\leq 2\lrA{1+\beta}\sum_{i=1}^{\rnum}\frac{w(s_i,t_i)}{\lambda}\leq\O(\beta)\cdot\OPT,
\]
where the first inequality follows from  $\Ceil{\frac{N_i}{\lambda}}\leq\frac{N_i}{\lambda}+1$ and $N_i>\frac{\lambda}{\beta}$ for each $(s^*_i,t^*_i)\in\P_{big}$, and the last from Lemma~\ref{thelb}.
Since $\beta=\Theta(\sqrt{\vertexnum\log\vertexnum})$, the total weight for serving all big requests is at most
{\begin{equation}\label{part2large}
w(\T) + \sum_{(s^*_i,t^*_i)\in\P_{big}}\Ceil{\frac{N_i}{\lambda}}\cdot 2w(s^*_i, t^*_i) \leq  \O\!\lrA{\sqrt{\vertexnum\log\vertexnum}}\cdot\OPT.
\end{equation}}

\textbf{Normal requests.} By definition, for each request $(s^*_i,t^*_i)\in\P_{normal}$, we have $\frac{\lambda}{\alpha}\leq N_i\leq\frac{\lambda}{\beta}$.
Hence, the total number of remaining requests in $\R$ is at most $\vertexnum^2\cdot\frac{\lambda}{\beta}$. We may assume $\lambda \gg \alpha$; otherwise, by Theorem~\ref{mainres1}, \textsc{Alg.1} satisfies all these requests with weight at most $\O(\sqrt{\alpha}\log\frac{\vertexnum^2\cdot\lambda}{\beta})\cdot\OPT\leq  \O(\sqrt{\vertexnum\log\vertexnum})\cdot\OPT$, as $\alpha=\Theta\lrA{\frac{\vertexnum}{\log\vertexnum}}$.

Let $I_u$ denote the unweighted \ac{mDaRP} instance consisting of all remaining requests, each with unit size.
We obtain a new unweighted \ac{mDaRP} instance $\widetilde{I}_u$ from $I_u$ via the following three steps.

\textbf{Step~1: Construct a weighted \ac{mDaRP} instance $\widehat{I}_w$.}
We scale down the vehicle capacity to $\widehat{\lambda}\coloneq \Floor{\frac{\lambda}{\gamma}} \cdot \gamma = \Floor{\alpha}\cdot\gamma$, where $\gamma\coloneq\frac{\lambda}{\alpha}$. 
For each request $(s^*_i,t^*_i)\in\P_{normal}$, we scale up its size to $\widehat{N}_i\coloneq \Ceil{\frac{N_i}{\gamma}}\cdot\gamma $. Given that $\frac{\lambda}{\alpha}\leq N_i\leq \frac{\lambda}{\beta}$, we have 
$
\widehat{N}_i\in\lrC{\gamma,2\cdot\gamma,\dots,\Ceil{\frac{\alpha}{\beta}}\cdot\gamma}.
$
Since $N_i\leq\widehat{N}_i\ll \widehat{\lambda}$, the instance $\widehat{I}_w$ is feasible.

\textbf{Step~2: Scale to an equivalent weighted \ac{mDaRP} instance $\widetilde{I}_w$.}
Dividing both the vehicle capacity and request sizes by $\gamma$ yields an equivalent weighted instance, with vehicle capacity $\widetilde{\lambda}\coloneq\frac{\widehat{N}_i}{\gamma}=\Floor{\alpha}$ and request sizes $\widetilde{N}_i\coloneq \frac{\widetilde{N}_i}{\gamma}\in\lrC{1,2,\dots,\Ceil{\frac{\alpha}{\beta}}}$.

\textbf{Step~3: Convert to an unweighted \ac{mDaRP} instance $\widetilde{I}_u$.} We construct an unweighted instance $\widetilde{I}_u$ by replacing each request $(s^*_i,t^*_i)\in\P_{normal}$ with $\widetilde{N}_i$ copies. 
Let $\widetilde{\R}$ denote the resulting (multi-)set of requests.
Then,  $\size{\widetilde{\R}}\leq \Ceil{\frac{\alpha}{\beta}}\cdot\vertexnum^2=\O(\vertexnum^{\frac{5}{2}})$, and the vehicle capacity remains $\widetilde{\lambda} = \Floor{\alpha} = \O\!\lrA{\frac{\vertexnum}{\log\vertexnum}}$.
By Theorem~\ref{mainres1}, \textsc{Alg.1} computes a solution of total weight
{\begin{align}\label{thetemweight}
\O\!\lrA{\sqrt{\frac{\vertexnum}{\log\vertexnum}}\log\vertexnum^{\frac{5}{2}}}\cdot\OPT(\widetilde{I}_u)&\leq\O\!\lrA{\sqrt{\vertexnum\log\vertexnum}}\cdot\OPT(\widetilde{I}_u)\leq \O\!\lrA{\sqrt{\vertexnum\log\vertexnum}}\cdot \OPT(\widetilde{I}_w),
\end{align}}
where the last inequality holds because $\widetilde{I}_u$ is a relaxation of $\widetilde{I}_w$, and thus $\OPT(\widetilde{I}_u)\leq \OPT(\widetilde{I}_w)$.

By extending the methods from \citep{DBLP:journals/talg/GuptaHNR10}, we obtain the following property for the \ac{mDaRP}.

\begin{lemma}\label{reduction1}
Consider a weighted \ac{mDaRP} instance $I$ with vehicle capacity $\lambda$, where each request size $N_i$ satisfies $N_i\leq\lambda$. Let $J$ be the corresponding unweighted instance obtained by replacing each request $(s^*_i,t^*_i)$ of size $N_i$ with $N_i$ unit requests.
Suppose $J$ admits a solution consisting of $r$ routes $\{W_1,\dots, W_r\}$ with total weight $\SOL \coloneq\sum_{j=1}^r w(W_j)$. 
Then, there exists a polynomial-time computable feasible solution for $I$ with weight at most $3\cdot \SOL$.
\end{lemma}
\begin{proof}
For each route $W_j$, construct a line (i.e., a path) $L_j$ by ``unfolding'' $W_j$ from its root as in~\citep{DBLP:journals/talg/GuptaHNR10}: every edge traversal of $W_j$ becomes one edge appended to $L_j$ with the same weight. Hence, $w(L_j) = w(W_j)$.
For any request $i$, i.e., $(s^*_i,t^*_i)$, the unweighted solution serves its $N_i$ unit copies, each traversing a segment from $s_i$ to $t_i$ on some route $W_j$. 
Therefore, request $i$ induces a collection of $N_i$ segments distributed among the lines $\{L_j\}_{j=1}^{r}$. 
Let $N_i^j$ denote the number of segments corresponding to request $i$ on $L_j$, and for each edge $e\in L_j$, let $N_i^{j,e}$ denote the number of these segments that contain $e$. 
Let $N_e^j\coloneq\sum_i N_i^{j,e}$ denote the total number of unweighted requests carried by $W_j$ when traversing edge $e$.
Since the unweighted solution is feasible, we must have $N_e^j \leq \lambda$ for all $j$ and $e\in L_j$.

\textbf{Random assignment.}
We randomly construct a weighted \ac{mDaRP} instance as follows: for each request $i$, choose uniformly at random one of its $N_i$ segments and assign the entire size $N_i$ to that segment.
For any edge $e\in L_j$, $\mathbb{P}\lrB{\text{request } i \text{ chooses a segment containing } e} = \frac{N_i^{j,e}}{N_i}$, so the random total size on $e$ is
$\sum_i N_i \cdot \mathbf{1}\{\text{request } i\text{ chooses a segment containing }e\}$,
and hence $\mathbb{E}[\text{size on } e]
= \sum_i N_i \cdot \frac{N_i^{j,e}}{N_i}
= \sum_i N_i^{j,e} = N_e^j \leq \lambda$.
Therefore, the expected flow lower bound on $L_j$ satisfies $\mathbb{E}[\mathrm{flowLB}(L_j)]
= \frac{1}{\lambda}\sum_{e\in L_j} w(e)\cdot \mathbb{E}[\text{size on } e] \leq \sum_{e\in L_j} w(e) = w( W_j)$.
Summing over all $j$ yields $\sum_{j=1}^r \mathbb{E}[\mathrm{flowLB}(L_j)] \leq \sum_{j=1}^r w(W_j) = \SOL$.

\textbf{Derandomization by minimum-weight segment selection and trimming.}
For each request $i$, let its $N_i$ candidate segments have weights $w_{i,1},\dots,w_{i,N_i}$.
If we assign the entire size $N_i$ to one segment of weight $w'$, then its contribution to the flow lower bound equals
$\frac{N_i}{\lambda}\cdot w'$.
Under the uniform random choice, the expected contribution equals
$\frac{N_i}{\lambda}\cdot \frac{1}{N_i}\sum_{i'=1}^{N_i}w_{i,i'}$.
Hence, choosing the segment with weight $w_{i}\coloneq\min_{i'\in [N_i]} w_{i,i'}$ yields a contribution
$\frac{N_i}{\lambda}\cdot w_{i} \leq \frac{N_i}{\lambda}\cdot \frac{1}{N_i}\sum_{i'}^{N_i} w_{i,i'}$,
which is no larger than the random expectation.
Summing over all requests (and lines) gives a deterministic assignment with $\sum_{j=1}^r \mathrm{flowLB}(L_j) \leq \sum_{j=1}^r \mathbb{E}[\mathrm{flowLB}(L_j)] \leq \SOL$.

After this assignment, the tail of a line $L_j$ may contain edges not covered by any chosen segment.
Removing these edges yields the \emph{trimmed line} $\widetilde{L}_j$.
Clearly, $\mathrm{flowLB}(\widetilde{L}_j) = \mathrm{flowLB}(L_j)$ and $w(\widetilde{L}_j) \leq w(L_j)$.
For a line metric, the Steiner lower bound equals the weight of the minimum-weight subpath covering all request endpoints and the vehicle, which is $w(\widetilde{L}_j)$. Thus,
$\sum_{j=1}^r \mathrm{SteinerLB}(\widetilde{L}_j)
= \sum_{j=1}^r w(\widetilde{L}_j) \leq \sum_{j=1}^r w(L_j) = \SOL$.

\textbf{Approximation on each line and lifting.}
By \citep{DBLP:journals/talg/GuptaHNR10}, the 3-approximation algorithm for the \ac{DaRP} on lines~\citep{krumke2000approximation} extends to weighted requests and is guaranteed
w.r.t.\ the preemptive flow and Steiner lower bounds. Applying this algorithm independently on each $\widetilde{L}_j$ yields a feasible route of weight at most
$3\cdot \max\lrC{\mathrm{flowLB}(\widetilde{L}_j),\mathrm{SteinerLB}(\widetilde{L}_j)}$.
Summing over $j$ and using $\sum_{j=1}^r\mathrm{flowLB}(\widetilde{L}_j)\leq \SOL$ and $\sum_{j=1}^r\mathrm{SteinerLB}(\widetilde{L}_j)\leq \SOL$, the total weight is at most $3\cdot\SOL$.
Since each $\widetilde{L}_j$ is a subpath of $L_j$, which in turn was obtained by unfolding $W_j$, mapping these routes back to the original metric yields a solution to the weighted \ac{mDaRP} instance $I$ while preserving the total weight.
All steps run in polynomial time.

Therefore, from the unweighted \ac{mDaRP} solution with total weight $\SOL$, we obtain a weighted \ac{mDaRP} solution of weight at most $3\cdot\SOL$ in polynomial time.
\end{proof}

\begin{corollary}\label{thesolution}
There is a polynomial-time algorithm that computes a solution to the weighted \ac{mDaRP} instance $\widehat{I}_w$ with weight at most $O\lrA{\sqrt{\vertexnum\log\vertexnum}}\cdot \OPT(\widehat{I}_w)$.
\end{corollary}
\begin{proof}
By \eqref{thetemweight} and Lemma~\ref{reduction1}, there is a polynomial-time algorithm that computes a solution to the weighted \ac{mDaRP} instance $\widetilde{I}_w$ with weight at most $O\lrA{\sqrt{\vertexnum\log\vertexnum}}\cdot \OPT(\widetilde{I}_w)$.

Since $\widehat{I}_w$ and $\widetilde{I}_w$ are equivalent encodings of the same instance, the same solution also applies to $\widehat{I}_w$, with weight at most $O\lrA{\sqrt{\vertexnum\log\vertexnum}}\cdot \OPT(\widetilde{I}_w) = O\lrA{\sqrt{\vertexnum\log\vertexnum}}\cdot \OPT(\widehat{I}_w)$.
\end{proof}

\begin{lemma}\label{part3normal}
There is a polynomial-time algorithm that computes a solution to the unweighted \ac{mDaRP} instance $I_u$ with weight at most $O\lrA{\sqrt{\vertexnum\log\vertexnum}}\cdot \OPT(I_u)\leq O\lrA{\sqrt{\vertexnum\log\vertexnum}}\cdot \OPT$.
\end{lemma}
\begin{proof}
It is clear that $\OPT(I_u)\leq \OPT$ by the triangle inequality. 
Moreover, the solution to the weighted instance $\widehat{I}_w$ from Corollary~\ref{thesolution} is also a feasible solution for $I_u$ with the same weight. Thus, it remains to show $\OPT(\widehat{I}_w)\leq \O(1)\cdot\OPT(I_u)$.

Let $I_w$ denote the weighted \ac{mDaRP} instance corresponding to $I_u$.
By Lemma~\ref{reduction1}, there exists a solution $\W=\{W_1,\dots,W_r\}$ to $I_w$ with weight at most $3\cdot\OPT(I_u)$, i.e., $\OPT(I_w)\leq 3\cdot\OPT(I_u)$.
Since $\widehat{I}_w$ rescales capacities and request sizes, the same routes $\W$ may not be directly feasible. 

Recall that $\lambda\gg\alpha$ and $N_i\leq\widehat{N}_i\ll \widehat{\lambda}$.
Therefore, by definition, we have 
$\widehat{\lambda} = \Floor{\alpha}\cdot\gamma = \Floor{\alpha}\cdot\frac{\lambda}{\alpha}\geq \lambda - \frac{\lambda}{\alpha}\gg \frac{\lambda}{2}$ and
$\widehat{N}_i = \Ceil{\frac{N_i}{\gamma}}\cdot\gamma\leq N_i + \gamma\leq 2N_i\ll \widehat{\lambda}$.
Hence, each $\widehat{N}_i$ is at most twice $N_i$, while $\widehat{\lambda}$ is at least $\frac{\lambda}{2}$.  
In particular, if the vehicles in $\widehat{I}_w$ had capacity $4\widehat{\lambda}\geq 2\lambda$, the routes $\W$ would be feasible.  

Since $N_i\ll \widehat{\lambda}$, we may assume that $\vertexnum$ is sufficiently large so that $\max_i N_i\leq\frac{\widehat{\lambda}}{3}$. 
Consequently, for each vehicle route $W_i \in \W$, we can use up to $4 + \Ceil{\frac{(4-1) \cdot \max_i N_i}{\widehat{\lambda}}} = 5$ vehicles with capacity $\widehat{\lambda}$ simultaneously to serve all requests assigned to $W_i$.
The reason is as follows. A vehicle with capacity $4\widehat{\lambda}$ can simulate the route $W_i$. 
If we instead use four vehicles each with capacity $\widehat{\lambda}$, then at most $4-1=3$ requests are split across two vehicles. 
Since $\max_i N_i\leq\frac{\widehat{\lambda}}{3}$, these requests can be handled by one more vehicle.
Hence, to realize the same service with only a single vehicle of capacity $\widehat{\lambda}$, we may let it traverse $W_i$ up to five times.
By the triangle inequality, this increases the total route weight to at most  $(2\cdot 5-1)\cdot w(W_i)=9\cdot w(W_i)$.

Therefore, there exists a solution to $\widehat{I}_w$ with weight at most $9\cdot\OPT(I_w)$, i.e., $\OPT(\widehat{I}_w)\leq 9\cdot\OPT(I_w)$. Since $\OPT(I_w)\leq 3\cdot\OPT(I_u)$, we conclude that $\OPT(\widehat{I}_w)\leq 27\cdot\OPT(I_u)$, as required.
\end{proof}

Combining \eqref{part1small}, \eqref{part2large}, and Lemma~\ref{part3normal}, we immediately obtain the following result.

\begin{theorem}
For the \ac{mDaRP}, there exists a polynomial-time $\O\!\lrA{\sqrt{\vertexnum\log\vertexnum}}$-approximation algorithm.
\end{theorem}

\section{Conclusion}
This paper investigates approximation algorithms for the \ac{mDaRP}. We propose an algorithm that achieves an approximation ratio of $\O(\sqrt{\vertexnum\log\vertexnum})$, even improving upon the previous best-known ratio of $\O(\sqrt{\vertexnum}\log^2\vertexnum)$ for the classic (single-vehicle) \ac{DaRP}. Our results are derived from new structural properties of the problem, which allow us to leverage techniques from the well-known \ac{CVRP}. We believe that these insights may also be applicable to other variants of the \ac{DaRP}. A key direction for future research is to determine whether the $\O(\sqrt{\vertexnum})$-approximation barrier can be surpassed.

\bibliographystyle{apalike}
\bibliography{main}

\begin{thebibliography}{}

\bibitem[Altinkemer and Gavish, 1987]{altinkemer1987heuristics}
Altinkemer, K. and Gavish, B. (1987).
\newblock Heuristics for unequal weight delivery problems with a fixed error guarantee.
\newblock {\em Oper. Res. Lett.}, 6(4):149--158.

\bibitem[Altinkemer and Gavish, 1990]{altinkemer1990heuristics}
Altinkemer, K. and Gavish, B. (1990).
\newblock Heuristics for delivery problems with constant error guarantees.
\newblock {\em Transp. Sci.}, 24(4):294--297.

\bibitem[Anily and Bramel, 1999]{anily1999approximation}
Anily, S. and Bramel, J. (1999).
\newblock Approximation algorithms for the capacitated traveling salesman problem with pickups and deliveries.
\newblock {\em Naval Research Logistics (NRL)}, 46(6):654--670.

\bibitem[Bartal, 1998]{DBLP:conf/stoc/Bartal98}
Bartal, Y. (1998).
\newblock On approximating arbitrary metrices by tree metrics.
\newblock In Vitter, J.~S., editor, {\em Proceedings of the Thirtieth Annual {ACM} Symposium on the Theory of Computing, Dallas, Texas, USA, May 23-26, 1998}, pages 161--168. {ACM}.

\bibitem[Blauth et~al., 2023]{blauth2022improving}
Blauth, J., Traub, V., and Vygen, J. (2023).
\newblock Improving the approximation ratio for capacitated vehicle routing.
\newblock {\em Math. Program.}, 197(2):451--497.

\bibitem[Chalasani and Motwani, 1999]{DBLP:journals/siamcomp/ChalasaniM99}
Chalasani, P. and Motwani, R. (1999).
\newblock Approximating capacitated routing and delivery problems.
\newblock {\em {SIAM} J. Comput.}, 28(6):2133--2149.

\bibitem[Charikar et~al., 2001]{DBLP:journals/siamcomp/CharikarKR01}
Charikar, M., Khuller, S., and Raghavachari, B. (2001).
\newblock Algorithms for capacitated vehicle routing.
\newblock {\em {SIAM} J. Comput.}, 31(3):665--682.

\bibitem[Charikar and Raghavachari, 1998]{DBLP:conf/focs/CharikarR98}
Charikar, M. and Raghavachari, B. (1998).
\newblock The finite capacity dial-a-ride problem.
\newblock In {\em 39th Annual Symposium on Foundations of Computer Science, {FOCS} 1998, Palo Alto, California, USA, November 8-11, 1998}, pages 458--467. {IEEE} Computer Society.

\bibitem[Christofides, 2022]{christofides1976worst}
Christofides, N. (2022).
\newblock Worst-case analysis of a new heuristic for the travelling salesman problem.
\newblock {\em Oper. Res. Forum}, 3(1).

\bibitem[Cordeau, 2006]{DBLP:journals/ior/Cordeau06}
Cordeau, J. (2006).
\newblock A branch-and-cut algorithm for the dial-a-ride problem.
\newblock {\em Oper. Res.}, 54(3):573--586.

\bibitem[Dantzig and Ramser, 1959]{dantzig1959truck}
Dantzig, G.~B. and Ramser, J.~H. (1959).
\newblock The truck dispatching problem.
\newblock {\em Manag. Sci.}, 6(1):80--91.

\bibitem[Erd{\"o}s and Szekeres, 1935]{erdos1935combinatorial}
Erd{\"o}s, P. and Szekeres, G. (1935).
\newblock A combinatorial problem in geometry.
\newblock {\em Compositio Mathematica}, 2:463--470.

\bibitem[Even et~al., 2004]{DBLP:journals/orl/EvenGKRS04}
Even, G., Garg, N., K{\"{o}}nemann, J., Ravi, R., and Sinha, A. (2004).
\newblock Min-max tree covers of graphs.
\newblock {\em Oper. Res. Lett.}, 32(4):309--315.

\bibitem[Fakcharoenphol et~al., 2004]{DBLP:journals/jcss/FakcharoenpholRT04}
Fakcharoenphol, J., Rao, S., and Talwar, K. (2004).
\newblock A tight bound on approximating arbitrary metrics by tree metrics.
\newblock {\em J. Comput. Syst. Sci.}, 69(3):485--497.

\bibitem[Frederickson et~al., 1976]{DBLP:conf/focs/FredericksonHK76}
Frederickson, G.~N., Hecht, M.~S., and Kim, C.~E. (1976).
\newblock Approximation algorithms for some routing problems.
\newblock In {\em 17th Annual Symposium on Foundations of Computer Science, Houston, Texas, USA, 25-27 October 1976}, pages 216--227. {IEEE} Computer Society.

\bibitem[Frederickson et~al., 1978]{DBLP:journals/siamcomp/FredericksonHK78}
Frederickson, G.~N., Hecht, M.~S., and Kim, C.~E. (1978).
\newblock Approximation algorithms for some routing problems.
\newblock {\em {SIAM} J. Comput.}, 7(2):178--193.

\bibitem[Friggstad et~al., 2022]{uncvrp}
Friggstad, Z., Mousavi, R., Rahgoshay, M., and Salavatipour, M.~R. (2022).
\newblock Improved approximations for capacitated vehicle routing with unsplittable client demands.
\newblock In Aardal, K.~I. and Sanit{\`{a}}, L., editors, {\em Integer Programming and Combinatorial Optimization - 23rd International Conference, {IPCO} 2022, Eindhoven, The Netherlands, June 27-29, 2022, Proceedings}, volume 13265 of {\em Lecture Notes in Computer Science}, pages 251--261. Springer.

\bibitem[Goemans and Williamson, 1995]{GoemansW95}
Goemans, M.~X. and Williamson, D.~P. (1995).
\newblock A general approximation technique for constrained forest problems.
\newblock {\em {SIAM} J. Comput.}, 24(2):296--317.

\bibitem[G{\o}rtz, 2006]{DBLP:conf/approx/Gortz06}
G{\o}rtz, I.~L. (2006).
\newblock Hardness of preemptive finite capacity dial-a-ride.
\newblock In D{\'{\i}}az, J., Jansen, K., Rolim, J. D.~P., and Zwick, U., editors, {\em Approximation, Randomization, and Combinatorial Optimization. Algorithms and Techniques, 9th International Workshop on Approximation Algorithms for Combinatorial Optimization Problems, {APPROX} 2006 and 10th International Workshop on Randomization and Computation, {RANDOM} 2006, Barcelona, Spain, August 28-30 2006, Proceedings}, volume 4110 of {\em Lecture Notes in Computer Science}, pages 200--211. Springer.

\bibitem[G{\o}rtz et~al., 2009]{DBLP:conf/esa/GortzNR09}
G{\o}rtz, I.~L., Nagarajan, V., and Ravi, R. (2009).
\newblock Minimum makespan multi-vehicle dial-a-ride.
\newblock In Fiat, A. and Sanders, P., editors, {\em Algorithms - {ESA} 2009, 17th Annual European Symposium, Copenhagen, Denmark, September 7-9, 2009. Proceedings}, volume 5757 of {\em Lecture Notes in Computer Science}, pages 540--552. Springer.

\bibitem[G{\o}rtz et~al., 2015]{DBLP:journals/talg/GortzNR15}
G{\o}rtz, I.~L., Nagarajan, V., and Ravi, R. (2015).
\newblock Minimum makespan multi-vehicle dial-a-ride.
\newblock {\em {ACM} Trans. Algorithms}, 11(3):23:1--23:29.

\bibitem[Gupta et~al., 2007]{DBLP:conf/esa/GuptaHNR07}
Gupta, A., Hajiaghayi, M., Nagarajan, V., and Ravi, R. (2007).
\newblock Dial a ride from \emph{k} -forest.
\newblock In Arge, L., Hoffmann, M., and Welzl, E., editors, {\em Algorithms - {ESA} 2007, 15th Annual European Symposium, Eilat, Israel, October 8-10, 2007, Proceedings}, volume 4698 of {\em Lecture Notes in Computer Science}, pages 241--252. Springer.

\bibitem[Gupta et~al., 2010]{DBLP:journals/talg/GuptaHNR10}
Gupta, A., Hajiaghayi, M.~T., Nagarajan, V., and Ravi, R. (2010).
\newblock Dial a ride from \emph{k}-forest.
\newblock {\em {ACM} Trans. Algorithms}, 6(2):41:1--41:21.

\bibitem[Haimovich and Kan, 1985]{HaimovichK85}
Haimovich, M. and Kan, A. H. G.~R. (1985).
\newblock Bounds and heuristics for capacitated routing problems.
\newblock {\em Math. Oper. Res.}, 10(4):527--542.

\bibitem[Ho et~al., 2018]{ho2018survey}
Ho, S.~C., Szeto, W.~Y., Kuo, Y.-H., Leung, J.~M., Petering, M., and Tou, T.~W. (2018).
\newblock A survey of dial-a-ride problems: Literature review and recent developments.
\newblock {\em Transportation Research Part B: Methodological}, 111:395--421.

\bibitem[Hu, 2009]{hu2009approximation}
Hu, Y. (2009).
\newblock {\em Approximation algorithms for the capacitated vehicle routing problem}.
\newblock PhD thesis, Simon Fraser University.

\bibitem[Karlin et~al., 2021]{DBLP:conf/stoc/KarlinKG21}
Karlin, A.~R., Klein, N., and Gharan, S.~O. (2021).
\newblock A (slightly) improved approximation algorithm for metric {TSP}.
\newblock In Khuller, S. and Williams, V.~V., editors, {\em {STOC} '21: 53rd Annual {ACM} {SIGACT} Symposium on Theory of Computing, Virtual Event, Italy, June 21-25, 2021}, pages 32--45. {ACM}.

\bibitem[Karlin et~al., 2023]{DBLP:conf/ipco/KarlinKG23}
Karlin, A.~R., Klein, N., and Gharan, S.~O. (2023).
\newblock A deterministic better-than-3/2 approximation algorithm for metric {TSP}.
\newblock In Pia, A.~D. and Kaibel, V., editors, {\em Integer Programming and Combinatorial Optimization - 24th International Conference, {IPCO} 2023, Madison, WI, USA, June 21-23, 2023, Proceedings}, volume 13904 of {\em Lecture Notes in Computer Science}, pages 261--274. Springer.

\bibitem[Karpinski et~al., 2015]{karpinski2015new}
Karpinski, M., Lampis, M., and Schmied, R. (2015).
\newblock New inapproximability bounds for {TSP}.
\newblock {\em J. Comput. Syst. Sci.}, 81(8):1665--1677.

\bibitem[Krumke et~al., 2000]{krumke2000approximation}
Krumke, S., Rambau, J., and Weider, S. (2000).
\newblock An approximation algorithm for the non-preemptive capacitated dial-a-ride problem.

\bibitem[Luo et~al., 2023]{DBLP:journals/pvldb/LuoFDG23}
Luo, K., Florio, A.~M., Das, S., and Guo, X. (2023).
\newblock A hierarchical grouping algorithm for the multi-vehicle dial-a-ride problem.
\newblock {\em Proc. {VLDB} Endow.}, 16(5):1195--1207.

\bibitem[Psaraftis, 1980]{psaraftis1980dynamic}
Psaraftis, H.~N. (1980).
\newblock A dynamic programming solution to the single vehicle many-to-many immediate request dial-a-ride problem.
\newblock {\em Transp. Sci.}, 14(2):130--154.

\bibitem[Segev and Segev, 2010]{DBLP:journals/algorithmica/SegevS10}
Segev, D. and Segev, G. (2010).
\newblock Approximate \emph{k}-steiner forests via the lagrangian relaxation technique with internal preprocessing.
\newblock {\em Algorithmica}, 56(4):529--549.

\bibitem[Stein, 1978]{stein1978scheduling}
Stein, D.~M. (1978).
\newblock Scheduling dial-a-ride transportation systems.
\newblock {\em Transp. Sci.}, 12(3):232--249.

\bibitem[Williamson and Shmoys, 2011]{williamson2011design}
Williamson, D.~P. and Shmoys, D.~B. (2011).
\newblock {\em The Design of Approximation Algorithms}.
\newblock Cambridge University Press.

\bibitem[Xu et~al., 2011]{DBLP:journals/orl/XuXR11}
Xu, Z., Xu, L., and Rodrigues, B. (2011).
\newblock An analysis of the extended christofides heuristic for the k-depot {TSP}.
\newblock {\em Oper. Res. Lett.}, 39(3):218--223.

\end{thebibliography}
\end{document}